\documentclass{article}

     \PassOptionsToPackage{numbers,compress,sort}{natbib}

     \usepackage[final]{neurips_2019}

\usepackage[utf8]{inputenc} 
\usepackage[T1]{fontenc}    
\usepackage{hyperref}       
\usepackage{url}            
\usepackage{booktabs}       
\usepackage{amsfonts}       
\usepackage{nicefrac}       
\usepackage{microtype}      
\usepackage{parskip}
\usepackage{amssymb,amsmath,amsthm,amsfonts}
\usepackage{mathtools}
\usepackage{enumitem}
\usepackage{graphicx}
\usepackage[font=small]{caption}
\usepackage{float}
\usepackage[ruled,vlined,linesnumbered]{algorithm2e}
\usepackage{footnote}
\usepackage{xcolor}

\usepackage{tikz}
\usetikzlibrary{calc,patterns,angles,quotes}
\usepackage{siunitx}

\usepackage{qcircuit}
\usepackage[labelformat=simple]{subcaption}

\newtheorem{theorem}{Theorem}

\newtheorem{lemma}{Lemma}

\numberwithin{equation}{section}
\numberwithin{theorem}{section}
\numberwithin{definition}{section}
\numberwithin{lemma}{section}
\numberwithin{proposition}{section}
\numberwithin{example}{section}
\numberwithin{corollary}{section}
\numberwithin{remark}{section}

\newcommand{\eq}[1]{(\ref{eq:#1})}

\newcommand{\rem}[1]{\hyperref[rem:#1]{Remark~\ref*{rem:#1}}}
\newcommand{\thm}[1]{\hyperref[thm:#1]{Theorem~\ref*{thm:#1}}}
\newcommand{\cor}[1]{\hyperref[cor:#1]{Corollary~\ref*{cor:#1}}}
\newcommand{\defn}[1]{\hyperref[defn:#1]{Definition~\ref*{defn:#1}}}
\newcommand{\lem}[1]{\hyperref[lem:#1]{Lemma~\ref*{lem:#1}}}
\newcommand{\prop}[1]{\hyperref[prop:#1]{Proposition~\ref*{prop:#1}}}
\newcommand{\fig}[1]{\hyperref[fig:#1]{Figure~\ref*{fig:#1}}}
\newcommand{\tab}[1]{\hyperref[tab:#1]{Table~\ref*{tab:#1}}}
\newcommand{\algo}[1]{\hyperref[algo:#1]{Algorithm~\ref*{algo:#1}}}
\renewcommand{\sec}[1]{\hyperref[sec:#1]{Section~\ref*{sec:#1}}}
\newcommand{\append}[1]{\hyperref[append:#1]{Supplemental Materials~\ref*{append:#1}}}
\newcommand{\fac}[1]{\hyperref[fac:#1]{Fact~\ref*{fac:#1}}}
\newcommand{\lin}[1]{\hyperref[lin:#1]{Line~\ref*{lin:#1}}}

\SetKwInput{KwInput}{Input}
\SetKwInput{KwOutput}{Output}

\def\>{\rangle}
\def\<{\langle}

\newcommand{\N}{\mathbb{N}}

\newcommand{\R}{\mathbb{R}}
\newcommand{\C}{\mathbb{C}}
\newcommand{\Q}{\mathbb{Q}}
\newcommand{\E}{\mathbb{E}}
\newcommand{\I}{\mathrm{I}}

\newcommand{\X}{\mathcal{X}}
\newcommand{\Y}{\mathcal{Y}}
\renewcommand{\P}{\mathcal{P}}
\renewcommand{\Q}{\mathcal{Q}}
\renewcommand{\H}{\mathcal{H}}

\newcommand{\D}{\mathcal{D}}
\renewcommand{\S}{\mathcal{S}}

\newcommand{\qW}{\mathrm{qW}}
\newcommand{\rank}{\mathrm{rank}}

\DeclareMathOperator{\poly}{poly}
\DeclareMathOperator{\sgn}{sgn}
\DeclareMathOperator{\SWAP}{\mathrm{SWAP}}

\renewcommand{\d}{\mathrm{d}}

\def\Tr{\operatorname{Tr}}
\newcommand{\range}[1]{[#1]}

\let\oldnl\nl
\newcommand{\nonl}{\renewcommand{\nl}{\let\nl\oldnl}}

\title{Quantum Wasserstein GANs}

\author{
  Shouvanik Chakrabarti$^{1,2,4,}$\thanks{Equal contribution.}, Yiming Huang$^{3,1,5,}$\footnotemark[1], Tongyang Li$^{1,2,4}$ \\
  \textbf{Soheil Feizi$^{2,4}$, Xiaodi Wu$^{1,2,4}$}\\
  $^1$ Joint Center for Quantum Information and Computer Science, University of Maryland \\
  $^2$ Department of Computer Science, University of Maryland \\
  $^3$  School of Information and Software Engineering\\
   University of Electronic Science and Technology of China \\
  $^4$ \texttt{\{shouv,tongyang,sfeizi,xwu\}@cs.umd.edu} \\
  $^5$ \texttt{yiminghwang@gmail.com}
}

\begin{document}
\maketitle
\begin{abstract}
The study of quantum generative models is well motivated, not only because of its importance in quantum machine learning and quantum chemistry but also because of the perspective of its implementation on near-term quantum machines. Inspired by previous studies on the adversarial training of classical and quantum generative models,  we propose the first design of quantum Wasserstein Generative Adversarial Networks (WGANs), which has been shown to improve the robustness and the scalability of the adversarial training of quantum generative models even on noisy quantum hardware.  Specifically, we propose a definition of the Wasserstein semimetric between quantum data, which inherits a few key theoretical merits of its classical counterpart.
We also demonstrate how to turn the quantum Wasserstein semimetric into a concrete design of quantum WGANs that can be efficiently implemented on quantum machines.
Our numerical study, via classical simulation of quantum systems, shows the more robust and scalable numerical performance of our quantum WGANs over other quantum GAN proposals.
As a surprising application, our quantum WGAN has been used to generate a 3-qubit quantum circuit of $\sim$50 gates that well approximates a 3-qubit 1-d Hamiltonian simulation circuit that requires over 10k gates using standard techniques.
\end{abstract}


\section{Introduction}
\label{sec:introduction}
\vspace{-1mm}
Generative adversarial networks (GANs) \cite{goodfellow2014generative} represent a power tool of training deep \emph{generative} models, which have a profound impact on machine learning. In GANs, a generator tries to generate fake samples resembling the true data, while a discriminator tries to discriminate between the true and the fake data. The learning process for generator and discriminator can be deemed as an adversarial game that converges to some equilibrium point under reasonable assumptions.

Inspired by the success of GANs and classical generative models, developing their quantum counterparts is a natural and important topic in the emerging field of quantum machine learning~\cite{schuld2015introduction, biamonte2017quantum}.
There are at least two appealing reasons for which quantum GANs are extremely interesting. First, quantum GANs could provide potential quantum speedups due to the fact that quantum generators and discriminators (i.e., parameterized quantum circuits) cannot be efficiently simulated by classical generators/discriminators. In other words, there might exist distributions that can be efficiently generated by quantum GANs, while otherwise impossible with classical GANs. Second, simple prototypes of quantum GANs (i.e., executing simple parameterized quantum circuits), similar to those of the variational methods (e.g.,~\cite{QAOA, NC-VQE, IBM-QE}), are likely to be implementable on near-term noisy-intermediate-scale-quantum (NISQ) machines~\cite{Preskill2018NISQ}. Since the seminal work of~\cite{lloyd18qgan}, there are quite a few proposals (e.g,~\cite{Killoran18qgan, benedetti2018adversarial, Zeng18, Situ18,hu2018quantum, Aspuru-Guzik19vqr, zoufal2019quantum}) of constructions of quantum GANs on how to encode quantum or classical data into this framework. Furthermore, \cite{hu2018quantum, zoufal2019quantum} also demonstrated proof-of-principle implementations of small-scale quantum GANs on actual quantum machines.

A lot of existing quantum GANs focus on using quantum generators to generate classical distributions. For truly quantum applications such as investigation of quantum systems in condensed matter physics or quantum chemistry, the ability to generate \emph{quantum data} is also important. In contrast to the case of classical distributions, where the loss function measuring the difference between the real and the fake distributions can be borrowed directly from the classical GANs, the design of the loss function between real and fake quantum data as well as the efficient training of the corresponding GAN is much more challenging. The only existing results on quantum data either have a unique design specific to the 1-qubit case~\cite{Killoran18qgan,hu2018quantum}, or suffer from robust training issues discussed below~\cite{benedetti2018adversarial}.

More importantly, classical GANs are well known for being delicate and somewhat unstable in training. In particular, it is known~\cite{arjovsky2017wasserstein} that the choice of the metric between real and fake distributions will be critical for the stability of the performance in the training. A few widely used ones such as the Kullback-Leibler (KL) divergence, the Jensen-Shannon (JS) divergence, and the total variation (or statistical) distance are not sensible for learning distributions supported by low dimensional generative models. The shortcoming of these metrics will likely carry through to their quantum counterparts and hence quantum GANs based on these metrics will likely suffer from the same weaknesses in training.
This training issue was not significant in the existing numerical study of quantum GANs in the 1-qubit case~\cite{Killoran18qgan,hu2018quantum}. However,  as observed by~\cite{benedetti2018adversarial} and us, the training issue becomes much more significant when the quantum system scales up, even just in the case of a few qubits.

To tackle the training issue of classical GANs, a lot of research has been conducted on the convergence of training GANs in classical machine learning.  A seminal work~\cite{arjovsky2017wasserstein} used \emph{Wasserstein distance} (or, \emph{optimal transport} distance)~\cite{villani2008optimal} as a metric for measuring the distance between real and fake distributions. Comparing to other measures (such as KL and JS), Wasserstein distance is more appealing from optimization perspective because of its continuity and smoothness. As a result, the corresponding Wasserstein GAN (WGAN) is promising for improving the training stability of GANs. There are a lot of subsequent studies on various modifications of the WGAN, such as GAN with regularized Wasserstein distance~\cite{sanjabi2018convergence}, WGAN with entropic regularizers~\cite{cuturi2013sinkhorn,seguy2017large}, WGAN with gradient penalty~\cite{gulrajani2017improved,petzka2017regularization}, relaxed WGAN~\cite{guo2017relaxed}, etc.
It is known~\cite{pmlr-v80-mescheder18a} that WGAN and its variants such as~\cite{gulrajani2017improved} have demonstrated improved training stability compared to the original GAN formulation.

\noindent \textbf{Contributions.\ \ }
Inspired by the success of classical Wasserstein GANs and the need of smooth, robust, and scalable training methods for quantum GANs on quantum data, we propose the first design of quantum Wasserstein GANs (qWGANs). To this end, our technical contributions are multi-folded.

In \sec{qW}, we propose a quantum counterpart of the Wasserstein distance, denoted by $\qW(P,Q)$ between quantum data $P$ and $Q$, inspired by~\cite{arjovsky2017wasserstein, villani2008optimal}. We prove that $\qW(\cdot, \cdot)$ is a semi-metric (i.e., a metric without the triangle inequality) over quantum data and inherits nice properties such as continuity and smoothness of the classical Wasserstein distance.
We will discuss a few other proposals of quantum Wasserstein distances such as~\cite{NGT15, Golse2016, peyr2016quantum, CARLEN20171810, chen2018matrix, chen2018wasserstein, yu2018quantum, qOT2019} and in particular why most of them are not suitable for the purpose of generating quantum data in GANs.
We will also discuss the limitation of our proposal of quantum Wasserstein semi-metric and hope its successful application in quantum GANs could provide another perspective and motivation to study this topic.

In \sec{qWGAN}, we show how to add the quantum \emph{entropic} regularization to $\qW(\cdot, \cdot)$ to further smoothen the loss function in the spirit of the classical case (e.g.,~\cite{sanjabi2018convergence}). We then show the construction of our regularized quantum Wasserstein GAN (qWGAN) in \fig{QWGAN-summary} and describe the configuration and the parameterization of both the generator and the discriminator. Most importantly, we show that the evaluation of the loss function and the evaluation of the gradient of the loss function can be in principle efficiently implemented on quantum machines.  This enables direct applications of classical training methods of GANs, such as alternating gradient-based optimization, to the quantum setting. It is a wide belief that classical computation cannot efficiently simulate quantum machines, in our case, the evaluation of the loss function and its gradient. Hence, the ability of evaluating them efficiently on quantum machines is \emph{critical} for its scalability.

In \sec{experiments}, we supplement our theoretical results with experimental validations via classical simulation of qWGAN. Specifically, we demonstrate numerical performance of our qWGAN for quantum systems up to 8 qubits for pure states and up to 3 qubits for mixed states (i.e., mixture of pure states). Comparing to existing results~\cite{Killoran18qgan,hu2018quantum, benedetti2018adversarial}, our numerical performance is more favorable in both the system size and its numerical stability.
To give a rough sense, a single step in the classical simulation of the 8-qubit system involves multiple multiplications of $2^8 \times 2^8$ matrices. Learning a mixed state is much harder than learning pure states (a reasonable classical analogue of their difference is the one between learning a Gaussian distribution and learning a mixture of Gaussian distributions~\cite{mixture-learning}).
We present the only result for learning a true mixed state up to 3 qubits.

Furthermore, following a specific 4-qubit generator that is recently implemented on an ion-trap quantum machine~\cite{ionq} and a reasonable noise model on the same machine~\cite{Zhu}, we simulate the performance of our qWGAN with noisy quantum operations.  Our result suggests that qWGAN can tolerant a reasonable amount of noise in quantum systems and still converge. This shows the possibility of implementing qWGAN on near-term (NISQ) machines~\cite{Preskill2018NISQ}.

Finally, we demonstrate a real-world application of qWGAN to approximate useful quantum application with large circuits by small ones. qWGAN can be used to approximate a potentially complicated unknown quantum state by a simple one when using a reasonably simple generator.
We leverage this property and the Choi-Jamio\l{}kowski isomorphism~\cite{nielsen2002quantum} between quantum operations and quantum states to generate a simple state that approximates another Choi-Jamio\l{}kowski state corresponding to potentially complicated circuits in real quantum applications.
The closeness in two Choi-Jamio\l{}kowski states of quantum circuits will translate to the average output closeness between two quantum circuits over random input states.
Specifically, we show that the quantum Hamiltonian simulation circuit for 1-d 3-qubit Heisenberg model in~\cite{childs2018towards} can be approximated by a circuit of 52 gates with an average output fidelity over 0.9999 and a worst-case error 0.15.
The best-known circuit based on the product formula will need $\sim$11900 gates, however, with a worst-case error 0.001.

\noindent \textbf{Related results.\ \ }
All existing quantum GANs~\cite{lloyd18qgan, Killoran18qgan, benedetti2018adversarial, Zeng18, Situ18,hu2018quantum, Aspuru-Guzik19vqr, zoufal2019quantum}, no matter dealing with classical or quantum data, have not investigated the possibility of using the Wasserstein distance. The most relevant works to ours are~\cite{Killoran18qgan, benedetti2018adversarial,hu2018quantum} with specific GANs dealing with quantum data. As we discussed above, ~\cite{Killoran18qgan,hu2018quantum} only discussed the 1-qubit case (both pure and mixed) and~\cite{benedetti2018adversarial} discussed the pure state case (up to 6 qubits) but with the loss function being the quantum counterpart of the total variation distance. Moreover, different from ours, the mixed state case in~\cite{Killoran18qgan} is a labeled one: in addition to observing their mixture, one also gets a label of which pure state it is sampled from.

\begin{figure}[!htb]
  \minipage[b]{0.48\textwidth}%
  \centering\hspace{0mm}
\Qcircuit @C=1em @R=1em {
   \lstick{\vec{e}_0}& \gate{\{(p_i,U_i)\}} & \measureD{\phi} &\multicgate{3}{L} \\
  \lstick{Q} & \qw & \measureD{\psi} & \cghost{L} \\
   \lstick{\vec{e}_0}& \gate{\{(p_i,U_i)\}} & \multimeasureD{1}{\xi_R} & \cghost{L}\\
  \lstick{Q} & \qw & \ghost{\xi_R} & \cghost{L}}
  \endminipage\hfill
  \minipage[b]{0.48\textwidth}%
  \centering\hspace{0mm}
\scalebox{0.9}{\Qcircuit @C=1em @R=1em {
   & \gate{R_{\sigma_1}(\theta_1)} & \multigate{1}{R_{\sigma_4}(\theta_4)} & \qw & \qw \\
   & \gate{R_{\sigma_2}(\theta_2)} & \ghost{R_{\sigma_4}(\theta_4)} & \multigate{1}{R_{\sigma_5}(\theta_5)} & \qw \\
  & \gate{R_{\sigma_3}(\theta_3)} & \qw & \ghost{R_{\sigma_4}(\theta_4)} & \qw \gategroup{1}{2}{3}{2}{0.7em}{--} \gategroup{1}{2}{3}{4}{0.7em}{--} \\}}
\endminipage\hfill

\minipage[t]{0.48\textwidth}
\centering
\caption*{(1) $\{(p_i, U_i)\}$ refers to the generator with initial state $\vec{e}_0$ and its parameterization; (2) $\phi, \psi, \xi_R$ refers to the discriminator; (3) the figure shows how to evaluate the loss function $L$ by measuring $\phi, \psi, \xi_R$ on the generated state and the real state $Q$ with post-processing.}
\endminipage\hfill
\minipage[t]{0.48\textwidth}
  \centering
  \caption*{An example of a parameterized 3-qubit quantum circuit for $U_i$ in the generator. $R_{\sigma_i}(\theta_i)=\exp(\frac{1}{2}\theta_i \sigma_i)$ denotes a Pauli rotation with angle $\theta_i$. It could be a 1-qubit or 2-qubit gate depending on the specific Pauli matrix $\sigma_i$. The circuit consists of many such gates.}
\endminipage
\caption{The Architecture of Quantum Wasserstein GAN.}
\label{fig:QWGAN-summary}
\end{figure}


\section{Classical Wasserstein Distance \& Wasserstein GANs}
\label{sec:classical}
Let us first review the definition of Wasserstein distance and how it is used in classical WGANs.

\noindent \textbf{Wasserstein distance\ \ }
Consider two probability distributions $p$ and $q$ given by corresponding density functions $p\colon \X \to \R, q\colon \Y \to \R$. Given a cost function $c \colon \mathcal{X} \times \mathcal{Y} \to \R$, the optimal transport cost between $p$ and $q$, known as the \emph{Kantorovich's} formulation~\cite{villani2008optimal}, is defined as
\begin{align}
  \label{eq:class-wass-primal}
  d_c(p,q) := \min_{\pi \in\Pi(p, q)} & \int_\mathcal{X}\int_\mathcal{Y} \pi(x,y)c(x,y) \,\d x\,\d y
\end{align}
where $\Pi(p,q)$ is the set of joint distributions $\pi$ having marginal distributions $p$ and $q$, i.e., $\int_\Y \pi(x,y)\,\d y = p(x)$ and $\int_\X \pi(x,y) \,\d x = q(y)$.

\noindent \textbf{Wasserstein GAN\ \ } The Wasserstein distance $d_c(p,q)$ can be used as an objective for learning a real distribution $q$ by a parameterized function $G_\theta$ that acts on a base distribution $p$. Then the objective becomes learning parameters $\theta$ such that $d_c(G_\theta(p),q)$ is minimized as follows:
\begin{align}
  \label{eq:class-wass-opt}
  \min_\theta \min_{\pi \in\Pi(\P,\Q)} & \int_\mathcal{X}\int_\mathcal{Y} \pi(x,y)c(G_\theta(x),y) \,\d x\,\d y.
\end{align}
  In \cite{arjovsky2017wasserstein}, Arjovsky et al. propose using the dual of \eq{class-wass-opt} to formulate the original min-min problem into a min-max problem, i.e., a generative adversarial network, with the following form:
\begin{align}
    \label{eq:class-wass-gan}
    \min_\theta \max_{\alpha,\beta} \quad & \E_{x \sim \P}[\phi_\alpha(x)] - \E_{y \sim \Q}[\psi_\beta(y)], \\
    \text{s.t } \quad & \phi_\alpha(G_\theta(x)) - \psi_\beta(y) \le c(G_\theta(x),y),\ \forall x,y, \label{eq:class-wass-dual-constraint}
\end{align}
where $\phi,\psi$ are functions parameterized by $\alpha,\beta$ respectively. This is advantageous because it is usually easier to parameterize functions rather than joint distributions. The constraint \eq{class-wass-dual-constraint} is usually enforced by a regularizer term for actual implementation. Out of many choices of regularizers, the most relevant one to ours is the entropy regularizer in~\cite{sanjabi2018convergence}. In the case that $c(x,y) = \lVert x-y \rVert_2$ and $\phi = \psi$ in \eq{class-wass-gan}, the constraint is that $\phi$ must be a $1$-Lipschitz function. This is often enforced by the gradient penalty method in a neural network used to parameterize $\phi$.


\section{Quantum Wasserstein Semimetric}
\label{sec:qW}

\noindent \textbf{Mathematical formulation of quantum data\ \ }
We refer curious readers to \append{prelim} for a more comprehensive introduction. Any quantum data (or quantum states) over space $\X$ (e.g., $\X=\C^d$) can be mathematically described by a \emph{density operator} $\rho$ that is a \emph{positive semidefinite} matrix (i.e., $\rho \succeq 0$) with trace one (i.e., $\Tr(\rho)=1$), and the set of which is denoted by $\D(\X)$.

A quantum state $\rho$ is \emph{pure} if $\rank{(\rho)}=1$; otherwise it is a \emph{mixed} state.
For a pure state $\rho$, it can be represented by the outer-product of a \emph{unit} vector $\vec{v} \in \C^d$, i.e., $\rho=\vec{v}\vec{v}^\dagger$, where $\dagger$ refers to conjugate transpose.
We can also use $\vec{v}$ to directly represent pure states. Mixed states are a classical mixture of pure states, e.g., $\rho=\sum_i p_i \vec{v_i}\vec{v_i}^\dagger$ where $p_i$s form a classical distribution and $\vec{v_i}$s are all unit vectors.

Quantum states in a composed system of $\X$ and $\Y$ are represented by density operators $\rho$ over the Kronecker-product space $\X \otimes \Y$ with dimension $\dim(\X)\dim(\Y)$. 1-qubit systems refer to $\X=\C^2$. A 2-qubit system has dimension 4 ($\X^{\otimes 2}$) and an $n$-qubit system has dimension $2^n$.
The partial trace operation $\Tr_\X(\cdot)$ (resp. $\Tr_\Y(\cdot)$) is a linear mapping from $\rho$ to its marginal state on $\Y$ (resp. $\X$).

\noindent \textbf{From classical to quantum data\ \ }  Classical distributions $p, q$ in \eq{class-wass-primal} can be viewed as special mixed states $\P \in \D(\X), \Q \in \D(\Y)$ where $\P, \Q$ are diagonal and $p, q$ (viewed as density vectors) are the diagonals of $\P$, $\Q$ respectively.  Note that this is different from the conventional meaning of samples from classical distributions, which are random variables with the corresponding distributions.

This distinction is important to understand quantum data as the former (i.e., density operators) rather than the latter (i.e., samples) actually represents the entity of quantum data. This is because there are multiple ways (different quantum measurements) to read out classical samples out of quantum data for one fixed density operator. Mathematically, this is because density operators in general can have off-diagonal terms and quantum measurements can happen along arbitrary bases.

Consider $\X$ and $\Y$ from \eq{class-wass-primal} being finite sets. We can express the classical Wasserstein distance \eq{class-wass-primal} as a special case of the matrix formulation of quantum data. Precisely, we can replace the integral in \eq{class-wass-primal} by summation, which can be then expressed by the trace of $\pi C$ where $C$ is a diagonal matrix with $c(x,y)$ in the diagonal. $\pi$ is also a diagonal matrix expressing the coupling distribution $\pi(x,y)$ of $p, q$. Namely, $\pi$'s diagonal is $\pi(x,y)$ and satisfies the coupling marginal condition $\Tr_\Y(\pi)=P$ and $\Tr_\X(\pi)=Q$ where $P,Q$ are diagonal matrices with the distribution of $p, q$ in the diagonal, respectively.  As a result,  the Kantorovich's optimal transport in \eq{class-wass-primal} can be reformulated as
\begin{align}\label{eq:matrix_Kantorovich}
  &d_c(p,q)  := \min_{\pi} \Tr(\pi C) \\
  \text{s.t.} \quad &\Tr_\Y (\pi) = \mathrm{diag}\{p(x)\},\ \Tr_\X (\pi )= \mathrm{diag}\{q(y)\},\ \pi \in \D(\X \otimes \Y), \nonumber
\end{align}
where $C=\mathrm{diag} \{c(x,y)\}$. Note that \eq{matrix_Kantorovich} is effectively a linear program.

\noindent \textbf{Quantum Wasserstein semimetric\ \ } Our matrix reformulation of the classical Wasserstein distance \eq{class-wass-primal} suggests a naive extension to the quantum setting as follows.  Let $\qW(\P, \Q)$ denote the quantum Wasserstein semimetric between $\P \in \D(\X), \Q \in \D(\Y)$, which is defined by
\begin{align}\label{eq:quant-wass-vanilla}
 &\qW(\mathcal{P},\mathcal{Q}) := \min_{\pi} \Tr(\pi C) \\
  \text{s.t.} \quad &\Tr_\Y (\pi) = \P,\ \Tr_\X (\pi )= \Q,\ \pi \in \D(\X \otimes \Y), \nonumber
\end{align}
where $C$ is a matrix over $\X \otimes \Y$ that should refer to some cost-type function. The choice of $C$ is hence critical to make sense of the definition. First, matrix $C$ needs to be Hermitian (i.e., $C=C^\dagger$) to make sure that $\qW(\cdot, \cdot)$ is real. A natural attempt is to use $C=\mathrm{diag}\{c(x,y)\}$ from \eq{matrix_Kantorovich}, which turns out to be significantly wrong.
This is because $\qW(\vec{v}\vec{v}^\dagger, \vec{v}\vec{v}^\dagger)$ will be strictly greater than zero for random choice of unit vector $\vec{v}$ in that case. This demonstrates a crucial difference between classical and quantum data: \emph{while classical information is always stored in the diagonal (or computational basis) of the space, quantum information can be stored off-diagonally (or in an arbitrary basis of the space)}. Thus, choosing a diagonal $C$ fails to detect the off-diagonal information in quantum data.

Our proposal is to leverage the concept of \emph{symmetric subspace} in quantum information~\cite{harrow13church} to make sure that $\mathrm{qW}(P,P)=0$ for any $P$. The projection onto the symmetric subspace is defined by
\begin{align}
\Pi_{\text{sym}}:=\frac{1}{2}(\I_{\X \otimes \Y} + \mathrm{SWAP}),
\end{align}
where $\I_{\X \otimes \Y}$ is the identity operator over $\X \otimes \Y$ and $\SWAP$ is the operator such that $\SWAP(\vec{x} \otimes \vec{y}) = (\vec{y} \otimes \vec{x}), \forall \vec{x} \in \X, \vec{y} \in \Y$.\footnote{One needs that $\X$ is isometric to $\Y$ to well define $\Pi_{\text{sym}}$. However, this is without loss of generality by choosing appropriate and potentially larger spaces $\X$ and $\Y$ to describe quantum data.}
It is well known that $\Pi_{\text{sym}}(\vec{u} \otimes \vec{u}) = \vec{u} \otimes \vec{u}$ for all unit vectors $u$. With this property and by choosing $C$ to be the complement of $\Pi_{\text{sym}}$, i.e.,
\begin{align}
C:=\I_{\X \otimes \Y} - \Pi_{\text{sym}} = \frac{1}{2}(\I_{\X \otimes \Y} - \SWAP),
\end{align}
we can show $\mathrm{qW}(P,P)=0$ for any $P$. This is achieved by choosing $\pi= \sum_i \lambda_i (\vec{v_i}\vec{v_i}^\dagger \otimes \vec{v_i}\vec{v_i}^\dagger)$ given $P$'s spectral decomposition $P= \sum_i \lambda_i \vec{v_i}\vec{v_i}^\dagger$. Moreover, we can show

\vspace{1mm}
\begin{theorem}[Proof in \append{qW-proof}]
    \label{thm:qwass-properties-main}
    $\qW(\cdot,\cdot)$ forms a semimetric over $\D(\X)$ over any space $\X$, i.e., for any $\P, \Q \in \D(\X)$,
    \begin{enumerate}
    \item $\qW(\P,\Q) \ge 0$,
    \item $\qW(\P, \Q)=\qW(\Q, \P)$,
    \item $\qW(\P,\Q) = 0$ iff $\P = \Q$.
     \end{enumerate}
 \end{theorem}

Even though our definition of $\mathrm{qW}(\cdot, \cdot)$, especially the choice of $C$, does not directly come from a cost function $c(x,y)$ over $\X$ and $\Y$,  it however still encodes some geometry of the space of quantum states. For example, let $P=\vec{v}\vec{v}^\dagger$ and $Q=\vec{u}\vec{u}^\dagger$, $\mathrm{qW}(P, Q)$ becomes 0.5 $(1-|\vec{u}^\dagger\vec{v}|^2)$ where $|\vec{u}^\dagger\vec{v}|$ depends on the angle between $\vec{u}$ and $\vec{v}$ which are unit vectors representing (pure) quantum states.

\noindent \textbf{The dual form of $\mathrm{qW}(\cdot, \cdot)$\ \ } The formulation of $\mathrm{qW}(\cdot, \cdot)$ in \eq{quant-wass-vanilla} is given by a semidefinite program (SDP), opposed to the classical form in \eq{matrix_Kantorovich} given by a linear program. Its dual form is as follows.
\begin{align}
  \label{eq:quant-wass-dual}
  \max_{\phi,\psi} \quad & \Tr(Q \psi) - \Tr(P \phi) \\
  \text{s.t.} \quad & \I_\X \otimes \psi - \phi \otimes \I_\Y \preceq C,
                  \phi \in \H(\X),\ \psi \in \H(\Y),  \nonumber
\end{align}
where $\H(\X) , \H(\Y)$ denote the set of Hermitian matrices over space $\X$ and $\Y$.
We further show the \emph{strong duality} for this SDP in \append{qW-proof}. Thus, both the primal \eq{quant-wass-vanilla} and the dual \eq{quant-wass-dual} can be used as the definition of $\qW(\cdot, \cdot)$.

\noindent \textbf{Comparison with other quantum Wasserstein metrics\ \ } There have been a few different proposals that introduce matrices into the original definition of classical Wasserstein distance. We will compare these definitions with ours and discuss whether they are appropriate in our context of quantum GANs.

A few of these proposals (e.g.,~\cite{Carlen2014, chen2018matrix,chen2018wasserstein}) extend the dynamical formulation of Benamou and Brenier~\cite{Benamou2000} in optimal transport  to the matrix/quantum setting. In this formulation, couplings are defined not in terms of joint density measures, but in terms of smooth paths $t \rightarrow \rho(x,t)$ in the space of densities that satisfy some continuity equation with some time dependent vector field $v(x,t)$ inspired by physics.  A pair $\{ \rho(\cdot, \cdot), v(\cdot, \cdot)\}$ is said to couple $P$ and $Q$, the set of which is denoted  $C(P, Q)$, if $\rho(x,t)$ is a smooth path with $\rho(\cdot,0)=P$ and $\rho(\cdot, 1)=Q$. The 2-Wasserstein distance is
\begin{align}
  \mathrm{W}_2(P, Q) = \inf_{ \{\rho(\cdot, \cdot), v(\cdot, \cdot)\} \in C(P,Q)} \frac{1}{2} \int_0^1 \int_{R^n} |v(x,t)|^2 \rho(x,t)\,\d x\,\d t.
\end{align}
The above formulation seems difficult to manipulate in the context of GAN. It is unclear (a) whether the above definition has a favorable duality to admit the adversarial training and (b) whether the physics-inspired quantities like $v(x,t)$ are suitable for the purpose of generating fake quantum data.

A few other proposals (e.g.,~\cite{NGT15, peyr2016quantum}) introduce the matrix-valued mass defined by a function $\mu : X \rightarrow C^{n \times n}$ over domain $X$, where $\mu(x)$ is  positive semidefinite and satisfies $\Tr(\int_X \mu(x) dx)=1$. Instead of considering transport probability masses from  $X$ to $Y$,  one considers transporting a matrix-valued mass $\mu_0(x)$ on $X$ to another matrix-valued mass $\mu_1(y)$ on $Y$.
One can similarly define the Kantorovich's coupling $\pi(x,y)$ of $\mu_0(x)$ and $\mu_1(y)$,  and define the Wasserstein distance based on a slight different combination of $\pi(x,y)$ and $c(x,y)$ comparing to \eq{class-wass-primal}. 
This definition, however, fails to derive a new metric between two matrices. This is because the defined Wasserstein distance still measures the distance between $X$ and $Y$ based on some induced measure ($\| \cdot \|_F$) on the dimension-$n$ matrix space.  This is more evident when $X= \{P\}$ and $Y=\{Q\}$. The Wasserstein distance reduces to $c(x,y)+ \| P-Q\|_F^2$ where  the Frobenius norm ($\| \cdot\|_F$) is directly used in the definition. 

The proposals in \cite{Golse2016, qOT2019} are very similar to us in the sense they define the same coupling in the Kantorovich's formulation as ours. However, their definition of the Wasserstein distance motivated by physics is induced by unbounded operator applied on continuous space, e.g., $\nabla_x$, $\mathrm{div}_x$. This makes their definition only applicable to continuous space, rather than qubits in our setting.

The closest result to ours is \cite{yu2018quantum}, although the authors haven't proposed one concrete quantum Wasserstein metric. Instead, they formulate a general form of reasonable quantum Wasserstein metrics between finite-dimensional quantum states and prove that Kantorovich-Rubinstein theorem does not hold under this general form. Namely, they show the trace distance between quantum states cannot be determined by any quantum Wasserstein metric out of their general form. 

\noindent \textbf{Limitation of our $\qW(\cdot, \cdot)$\ \ } Although we have successfully implemented qWGAN based on our $\qW(\cdot, \cdot)$ and observed improved numerical performance, there are a few perspectives about  $\qW(\cdot, \cdot)$ worth further investigation. First, numerical study reveals that $\qW(\cdot, \cdot)$ does not satisfy the triangle inequality. 
Second, our $\qW(\cdot, \cdot)$ does not come from an explicit cost function, even though it encodes some geometry of the quantum state space. 
We conjecture that there could be a concrete underlying cost function and our $\qW(\cdot, \cdot)$ (or a related form) could be emerged as the 2-Wasserstein metric of that cost function. 
We hope our work provides an important motivation to further study this topic.


\section{Quantum Wasserstein GAN}
\label{sec:qWGAN}
We describe the specific architecture of our qWGAN (\fig{QWGAN-summary}) and its training. Similar to \eq{class-wass-opt} with the fake state $P$ from a parameterized quantum generator $G$, consider
\begin{align}\label{eq:quant-wgan-primal}
 \min_G \min_{\pi} \quad &\Tr(\pi C) \\
  \text{s.t.} \quad &\Tr_\Y (\pi) = P,   \Tr_\X (\pi )= Q ,  \pi \in \D(\X \otimes \Y), \nonumber
\end{align}
or similar to \eq{class-wass-gan} by taking the dual from \eq{quant-wass-dual},
\begin{align}
  \label{eq:quant-wgan}
  \min_G\max_{\phi,\psi} \quad & \Tr(Q \psi) - \Tr(P \phi) = \E_Q[\psi] - \E_P[\phi]\\
  \text{s.t.} \quad & \I_\X \otimes \psi - \phi \otimes \I_\Y \preceq C,
                  \phi \in \H(\X),\ \psi \in \H(\Y), \nonumber
\end{align}
where we abuse the notation of $\E_Q[\psi]\vcentcolon= \Tr(Q \psi)$, which refers to the expectation of the outcome of measuring Hermitian $\psi$ on quantum state $Q$.
We hence refer $\phi,\psi$ as the discriminator.

\vspace{-1mm}
\subsection*{Regularized Quantum Wasserstein GAN}\vspace{-1mm}
\label{sec:regul-quant-wass}
The dual form \eq{quant-wgan} is inconvenient for optimizing directly due to the constraint $\I_\X \otimes \psi - \phi \otimes \I_\Y \preceq C$.
Inspired by the entropy regularizer in the classical setting (e.g.,~\cite{sanjabi2018convergence}), we add a \emph{quantum-relative-entropy-based} regularizer
between $\pi$ and $P \otimes Q$ with a tunable parameter $\lambda$ to \eq{quant-wgan-primal} to obtain
\begin{align}
  \label{eq:qwgan-primal-regularized}
  \min_G \min_{\pi} \quad &\Tr(\pi C) + \lambda\Tr(\pi \log(\pi) - \pi\log(P \otimes Q)) \\
  \text{s.t.} \quad &\Tr_\Y (\pi) = P,   \Tr_\X (\pi )= Q ,  \pi \in \D(\X \otimes \Y). \nonumber
\end{align}
Using duality and the Golden-Thomposon inequality~\cite{golden1965lower,thompson1965inequality}, we can approximate \eq{qwgan-primal-regularized} by
\begin{align}
  \label{eq:qwgan-regularized}
\min_G  \max_{\phi,\psi} \quad  \E_Q[\psi] - \E_P [\phi] - \E_{P \otimes Q}[\xi_R]
\quad   \text{s.t.}\ \  \phi \in \H(\X),\ \psi \in \H(\Y),
\end{align}
where $\xi_R$ refers to the regularizing Hermitian
\begin{align}
  \xi_R= \frac{\lambda}{e}\exp\left(\frac{ -C - \phi \otimes \I_\Y +\I_\X \otimes \psi}{\lambda} \right).
\end{align}
Similar to~\cite{sanjabi2018convergence}, we prove that this entropic regularization ensures that the objective for the outer minimization problem \eq{qwgan-regularized} is \emph{differentiable} in $P$. (Proofs are given in \append{regul-quant-wass-appendix}.)

\vspace{-1mm}
\subsection*{Parameterization of the Generator and the Discriminator}
\label{sec:param-gener-discr}
\vspace{-2mm}

\noindent \textbf{Generator $G$} is a quantum operation that generates $P$ from a fixed initial state $\rho_0$ (e.g., the classical all-zero state $\vec{e}_0$). Specifically, generator $G$ can be described by an ensemble $\{(p_1,U_1),\dots,(p_r,U_r)\}$ that means applying the unitary $U_i$ with probability $p_i$. The distribution $\{p_1, \ldots, p_r\}$ can be parameterized directly or through some classical generative network. The rank of the generated state is $r$ ($r=1$ for pure states and $r>1$ for mixed states). Our experiments include the cases $r=1,2$.

Each unitary $U_i$ refers to a quantum circuit consisting of simple parameterized 1-qubit and 2-qubit Pauli-rotation quantum gates (see the right of \fig{QWGAN-summary}). These Pauli gates can be implemented on near-term machines (e.g.,~\cite{ionq}) and also form a universal gate set for quantum computation. Hence, this generator construction is widely used in existing quantum GANs. The $j$th gate in $U_i $ contains an angle $\theta_{i,j}$ as the parameter. All variables $p_i$, $\theta_{i,j}$ constitute the set of parameters for the generator.

\noindent \textbf{Discriminator $\phi, \psi$} can be parameterized at least in two ways. The first approach is to represent $\phi, \psi$ as linear combinations of tensor products of Pauli matrices, which form a basis of the matrix space (details on Pauli matrices and measurements can be found in \append{prelim}). Let $\phi = \sum_k \alpha_k A_k$ and $\psi = \sum_l \beta_l B_l$, where $A_k,B_l$ are tensor products of Pauli matrices. To evaluate $\E_P[\phi]$ (similarly for $\E_Q[\psi]$), by linearity it suffices to collect the information of $\E_P[A_k]$s,  which are simply Pauli measurements on the quantum state $P$ and amenable to experiments. Hence, $\alpha_k$ and $\beta_l$ can be used as the parameters of the discriminator.
The second approach is to represent $\phi, \psi$ as parameterized quantum circuits (similar to the $G$) with a measurement in the computational basis. The set of parameters of $\phi$ (respectively $\psi$) could be the parameters of the circuit and values associated with each measurement outcome. Our implementation mostly uses the first parameterization.

\vspace{-2mm}
\subsection*{Training the Regularized Quantum Wasserstein GAN}
\label{sec:train-regul-quant}
\vspace{-2mm}

For the scalability of the training of the Regularized Quantum Wasserstein GAN, one must be able to evaluate the loss function $L = \E_Q[\psi] - \E_P[\phi] - \E_{P \otimes Q}[\xi_R]$ or its gradient efficiently on a quantum computer. Ideally, one would hope to directly approximate gradients by quantum computers to facilitate the training of qWGAN, e.g., by using the alternating gradient descent method.  We show that it is indeed possible and outline the key steps.  More details are in \append{impl-disc-phys}.

\noindent \textbf{Computing the loss function:} Each unitary operation $U_i$ that refers to an actual quantum circuit can be efficiently evaluated on quantum machines in terms of the circuit size.
It can be shown that $L$ is a linear function of $P$ and can be computed by evaluating each $L_i=\E_Q[\psi] - \E_{U_i\rho_0 U_i^\dagger}[\phi] - \E_{U_i\rho_0 U_i^\dagger \otimes Q}[\xi_R]$
where $U_i\rho_0 U_i^\dagger$ refers to the state after applying $U_i$ on $\rho_0$. Similarly, one can show that $L$ is a linear function of the Hermitian matrices $\phi,\psi,\xi_R$. Our parameterization of $\phi$ and $\psi$ readily allows the use of efficient Pauli measurements to evaluate $\E_P[\phi]$ and $\E_Q[\psi]$. To handle the tricky part $\E_{P \otimes Q}[\xi_R]$, we relax $\xi_R$ and use a Taylor series to approximate $\E_{P \otimes Q}[\xi_R]$; the result form can again be evaluated by Pauli measurements composed with simple SWAP operations. As the major computation (e.g., circuit evaluation and Pauli measurements) is efficient on quantum machines, the overall implementation is efficient with possible overhead of sampling trials.

\noindent \textbf{Computing the gradients:} The parameters of the qWGAN are $\{p_i\} \cup \{\theta_{i,j}\} \cup \{\alpha_k\} \cup \{\beta_l\}$. $L$ is a linear function of $p_i,\alpha_k,\beta_l$. Thus it can be shown that the partial derivatives w.r.t. $p_i$ can be computed by evaluating the loss function on a generated state $U_i\rho_0 U_i^\dagger$ and the partial derivatives w.r.t. $\alpha_k,\beta_l$ can be computed by evaluating the loss function with $\phi,\psi$ replaced with $A_k,B_l$ respectively. The partial derivatives w.r.t. $\theta_{i,j}$ can be evaluated using techniques due to~\cite{schuld2019evaluating} via a simple yet elegant modification of the quantum circuits used to evaluate the loss function. The complexity analysis is similar to above. The only new ingredient is the quantum circuits to evaluate the partial derivatives w.r.t. $\theta_{i,j}$ due to~\cite{schuld2019evaluating}, which are again efficient on quantum machines.

\noindent \textbf{Summary of the training complexity:} A rough complexity analysis above suggests that one step of the evaluation of the loss function (or the gradients) of our qWGAN can be efficiently implemented on quantum machines. (A careful analysis is in \append{comp-cost-eval}.) Given this ability, the rest of the training of qWGAN is similar to the classical case and will share the same complexity. It is worthwhile mentioning that quantum circuit evaluation and Pauli measurements are not known to be efficiently computable by classical machines; the best known approach will cost exponential time.


\section{Experimental Results}
\label{sec:experiments}

We supplement our theoretical findings with numerical results by classical simulation of quantum WGANs of learning \emph{pure} states (up to 8 qubits) and \emph{mixed} states (up to 3 qubits) as well as its performance on noisy quantum machines. We use quantum fidelity between the generated and target states to track the progress of our quantum WGAN.
If the training is successful, the fidelity will approach $1$. Our quantum WGAN is trained using the alternating gradient descent method.

In most of the cases, the target state is generated by a circuit sharing the same structure with the generator but with randomly chosen parameters. We also demonstrate a special target state corresponding to useful quantum unitaries via the Choi-Jamio\l{}kowski isomorphism.
More details of the following experiments (e.g., parameter choices) can be found in \append{numerical-results}.

Most of the simulations were run on a dual core Intel I5 processor with 8G memory. The $8$-qubit pure state case was run on a Dual Intel Xeon E5-2697 v2 @ 2.70GHz processor with 128G memory. All source codes are publicly available at \url{https://github.com/yiminghwang/qWGAN}.

\noindent \textbf{Pure states\ \ } We demonstrate a typical performance of quantum WGAN of learning $1, 2, 4$, and $8$ qubit pure states in \fig{pure-state-clean}. We also plot the average fidelity for 10 runs with random initializations in \fig{average-pure} which shows the numerical stability of qWGAN.

\begin{figure}[!tb]
  \minipage{0.48\textwidth}%
  \centering
  \includegraphics[width=0.7\linewidth, height=2.8cm]{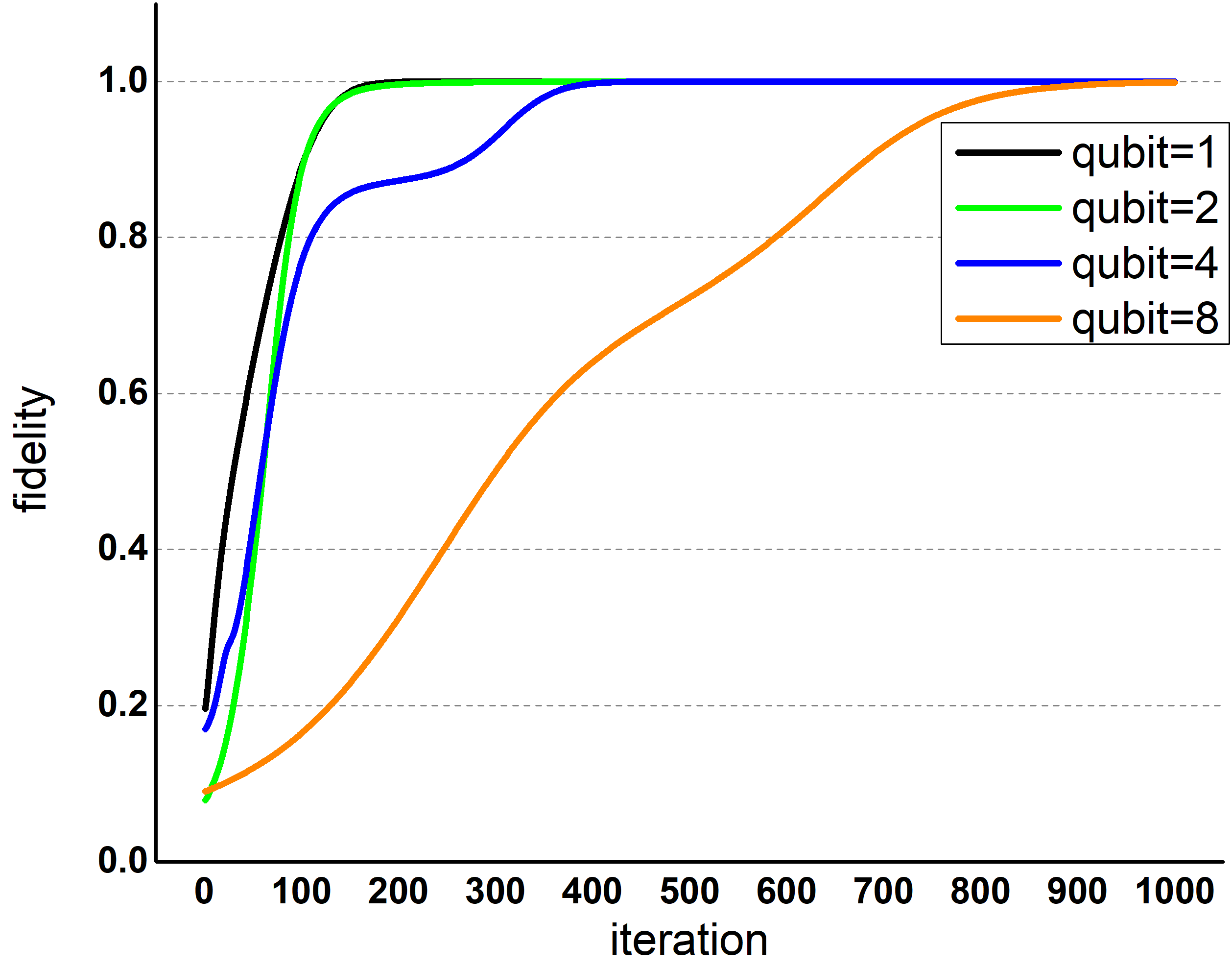}
  \caption*{Fidelity vs Training Epochs}
  \endminipage\hfill
  \minipage{0.48\textwidth}%
  \centering
  \includegraphics[width=0.7\linewidth]{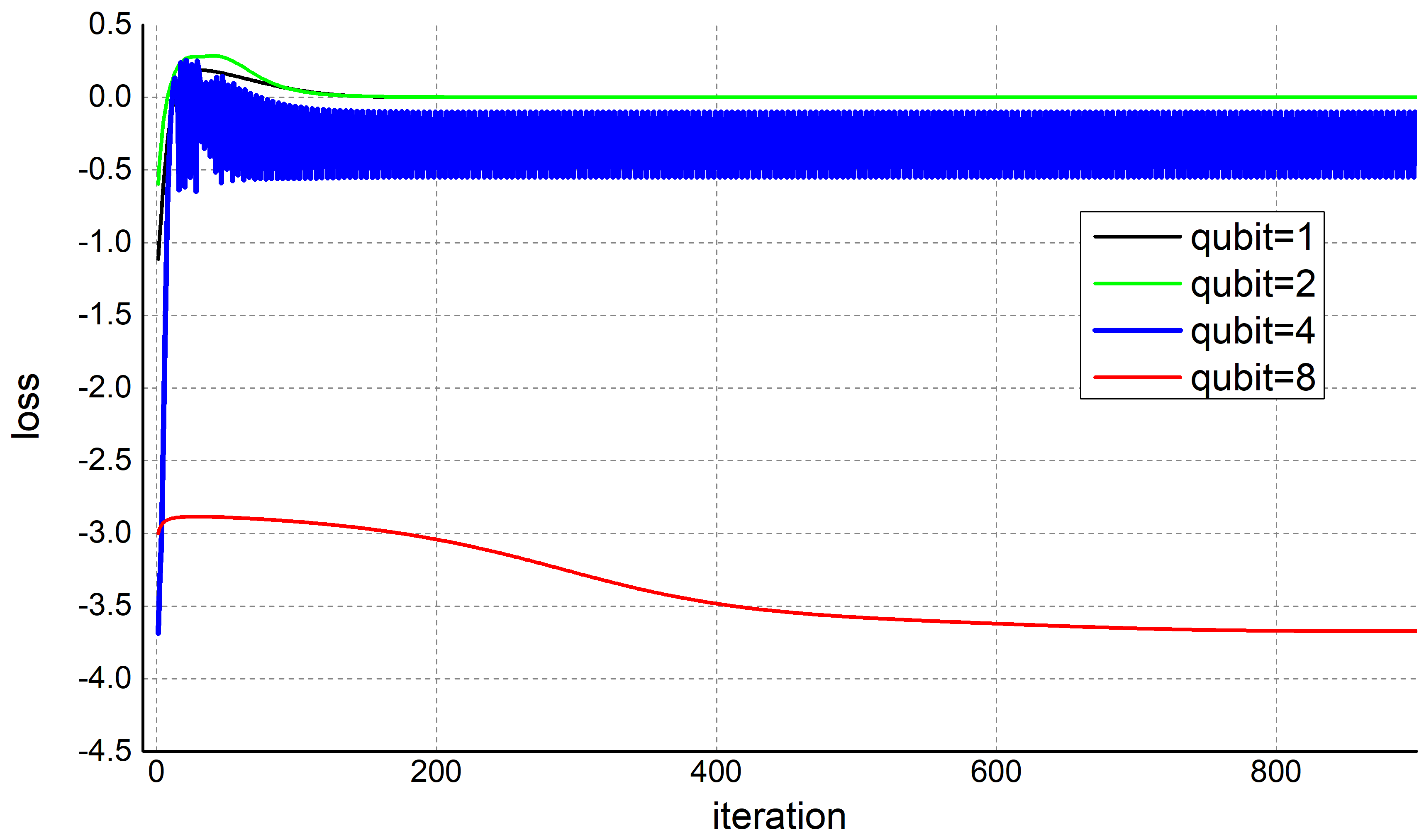}
  \caption*{Training Loss}
  \endminipage\hfill
  \caption{A typical performance of learning pure states (1,2,4, and 8 qubits).}
  \label{fig:pure-state-clean}
\end{figure}

\begin{figure}[!htb]
  \begin{minipage}[b]{0.25\textwidth}
  \includegraphics[width=\linewidth, height=2.5cm]{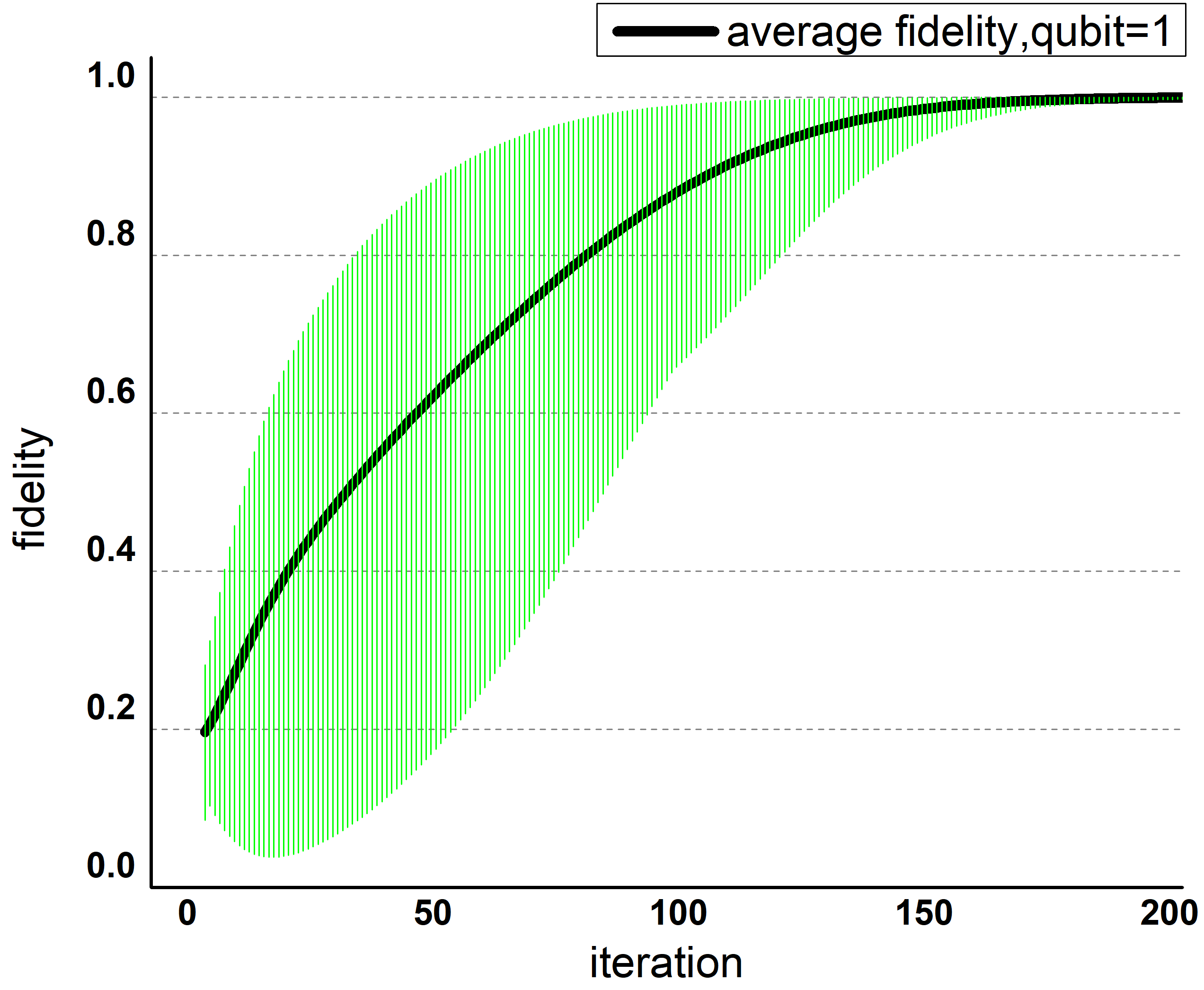}
  \caption*{1 qubit}
  \end{minipage}\hfill
\begin{minipage}[b]{0.25\textwidth}
  \includegraphics[width=\linewidth, height=2.5cm]{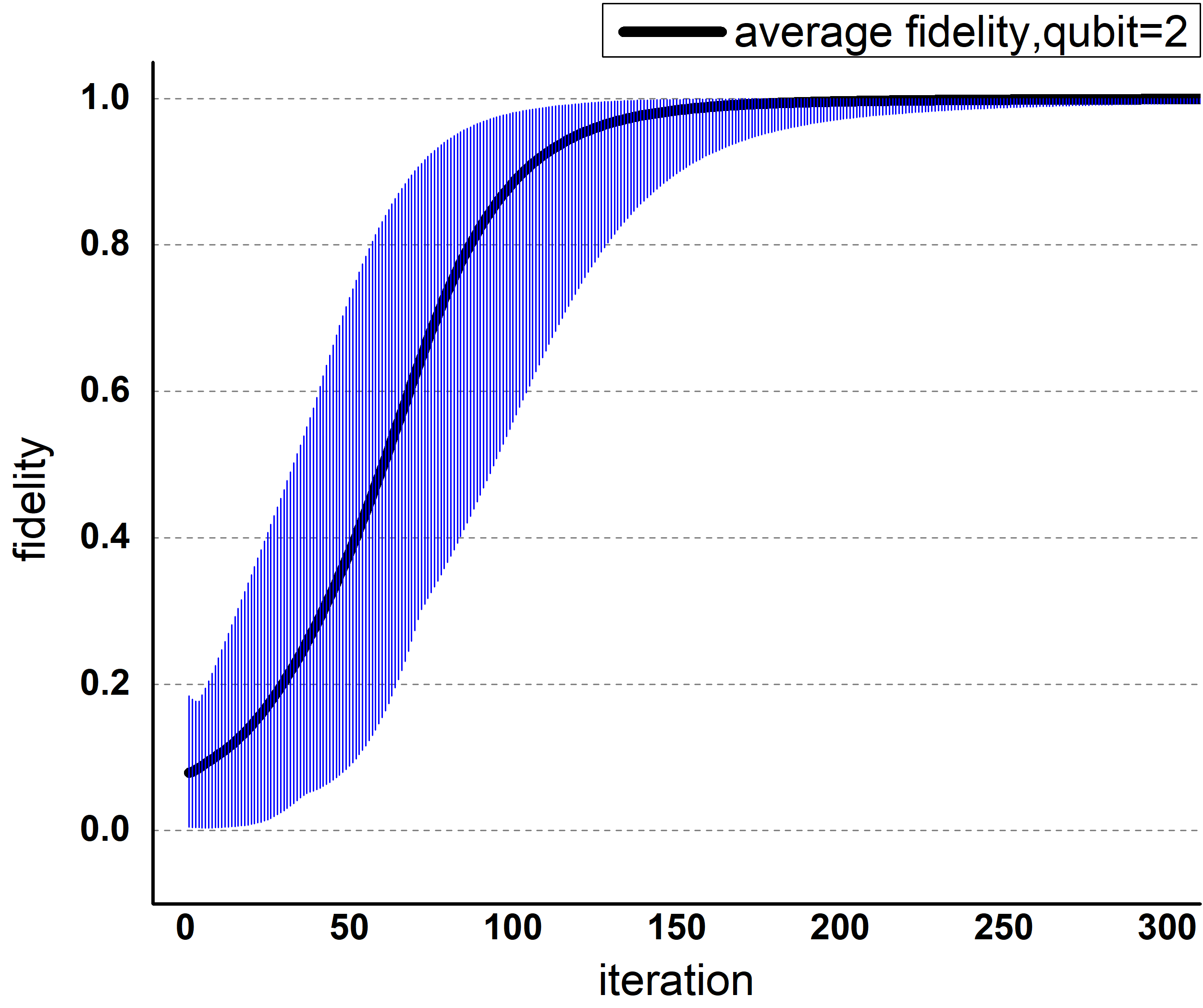}
  \caption*{2 qubits}
\end{minipage}\hfill
\begin{minipage}[b]{0.25\textwidth}%
  \includegraphics[width=\linewidth, height=2.5cm]{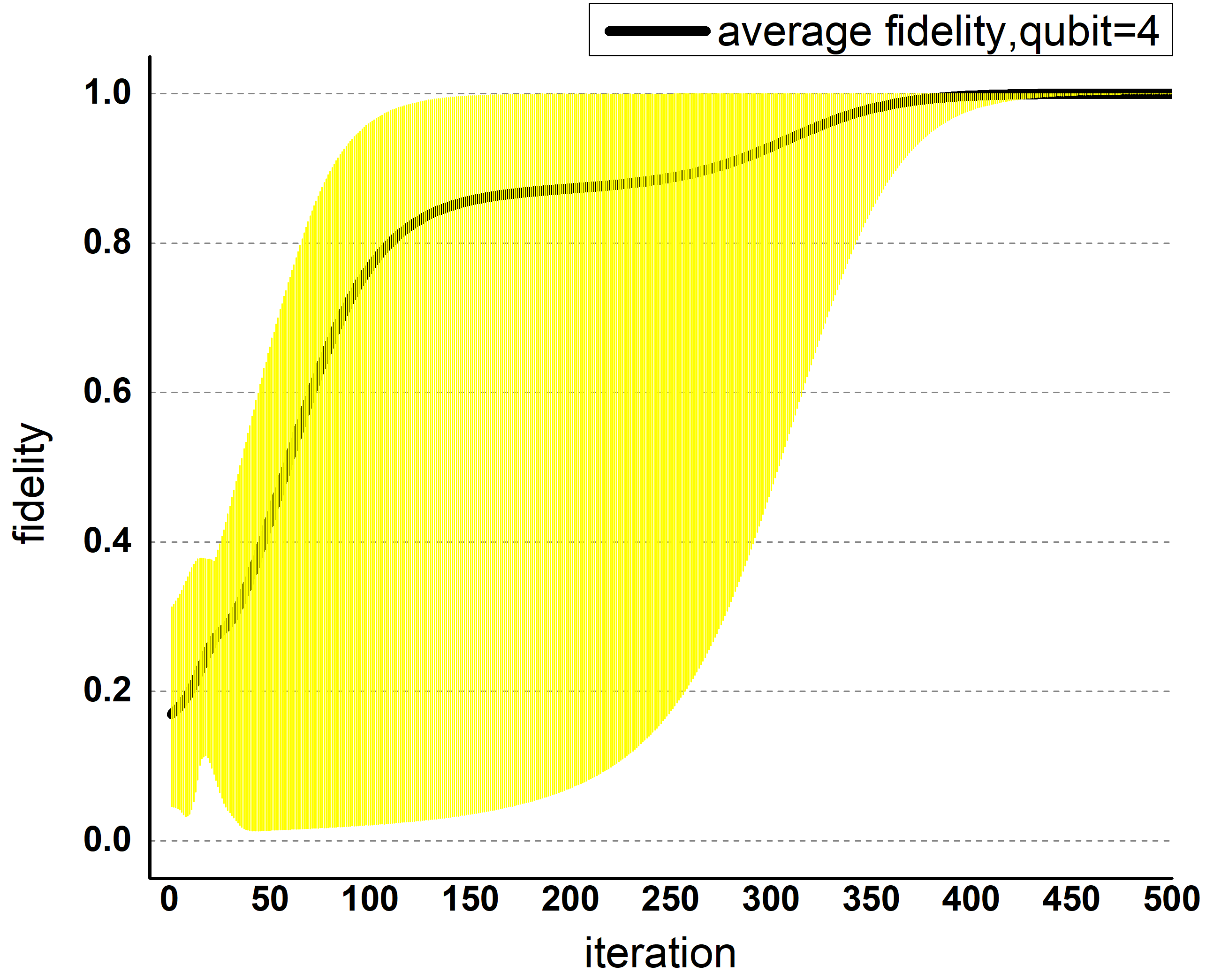}
  \caption*{4 qubits}
\end{minipage}\hfill
\begin{minipage}[b]{0.25\textwidth}%
  \includegraphics[width=\linewidth, height=2.5cm]{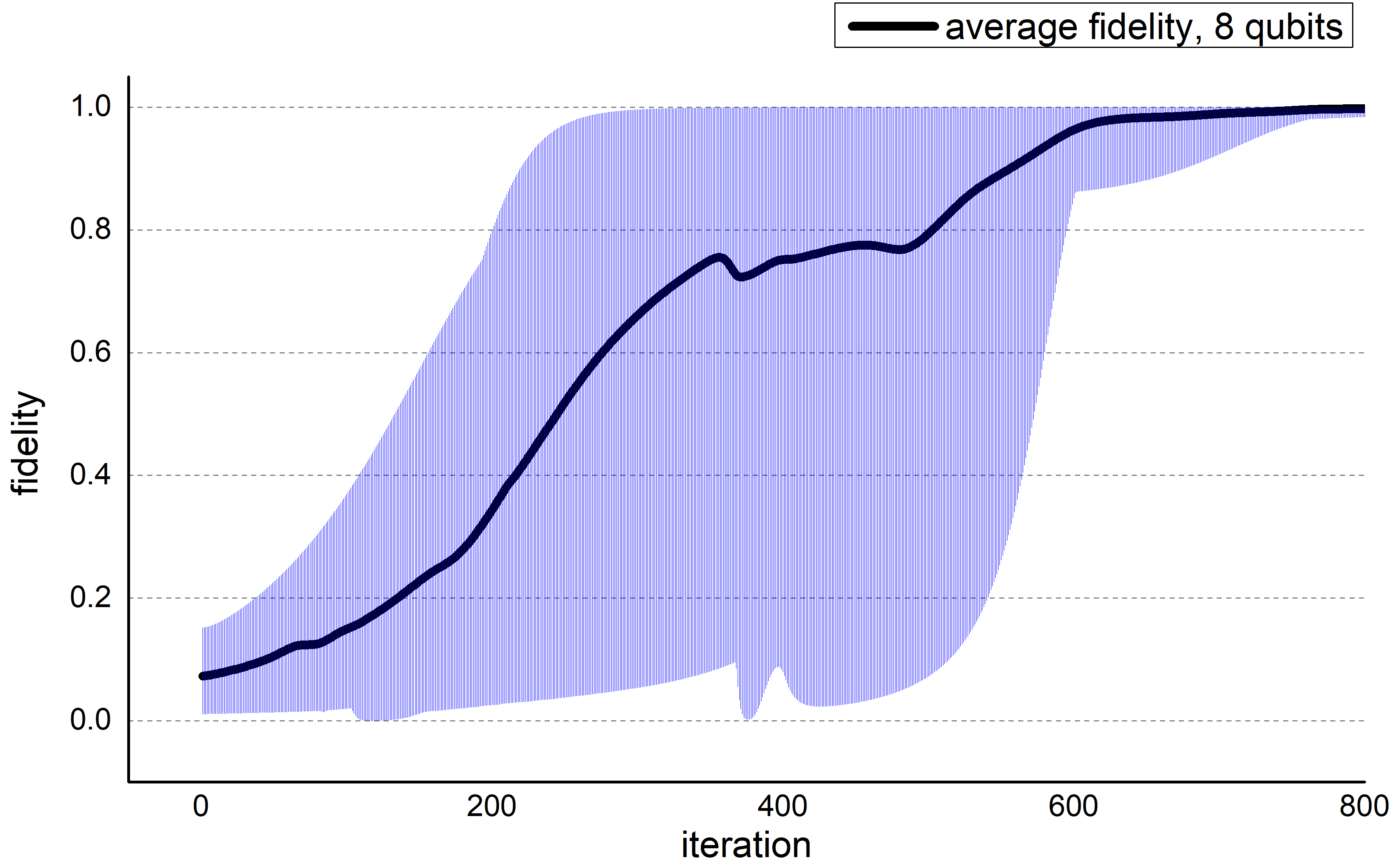}
  \caption*{8 qubits}
\end{minipage}
  \caption{Average performance of learning pure states (1, 2, 4, 8 qubits) where the black line is the average fidelity over multi-runs with random initializations and the shaded area refers to the range of the fidelity.}
  \label{fig:average-pure}
\end{figure}

\begin{figure}[!htb]
  \begin{minipage}[b]{0.32\textwidth}
  \includegraphics[width=\linewidth, height=2.5cm]{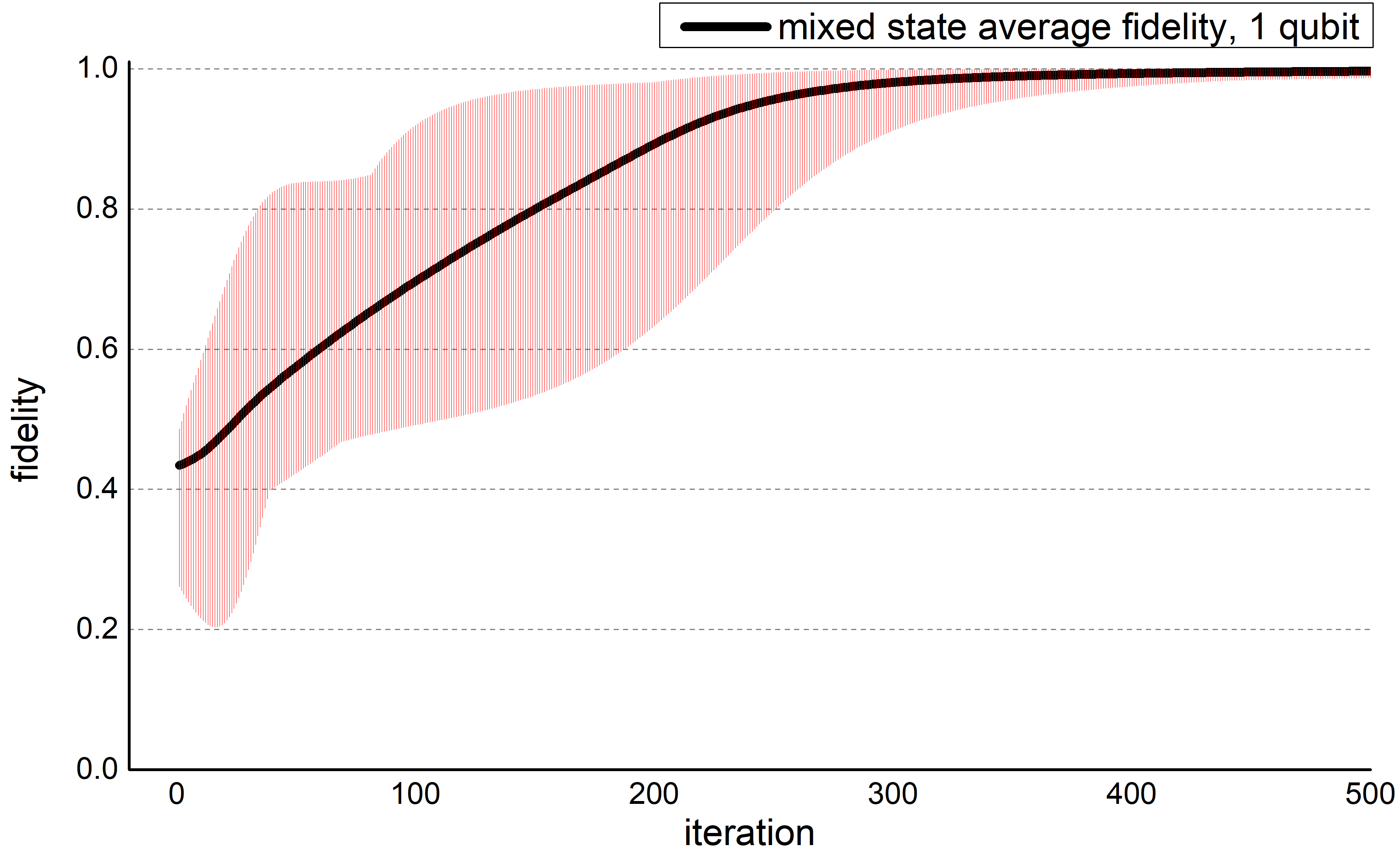}
  \caption*{1 qubit}
  \end{minipage}\hfill
\begin{minipage}[b]{0.32\textwidth}
  \includegraphics[width=\linewidth, height=2.5cm]{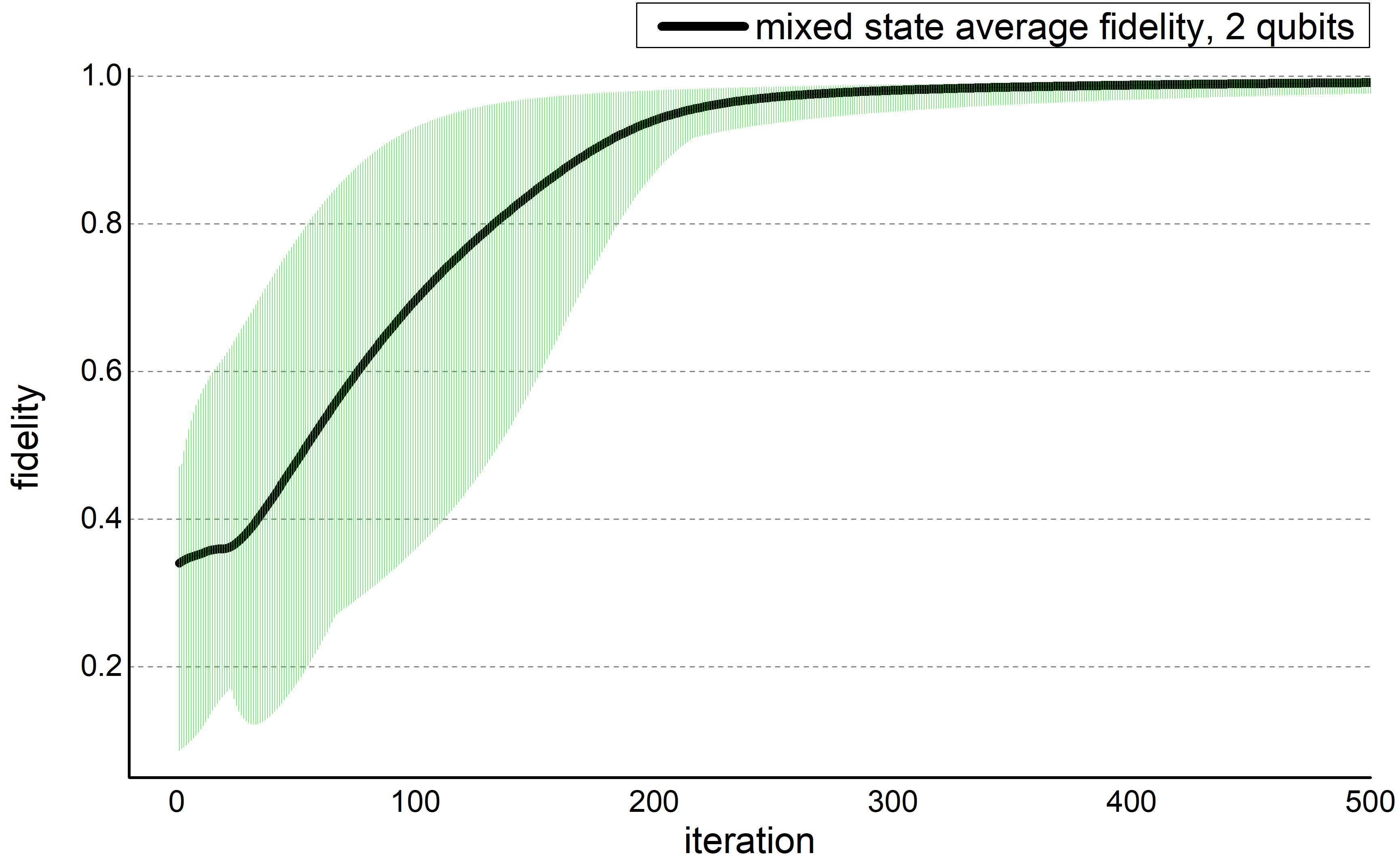}
  \caption*{2 qubits}
\end{minipage}\hfill
\begin{minipage}[b]{0.32\textwidth}%
  \includegraphics[width=\linewidth, height=2.5cm]{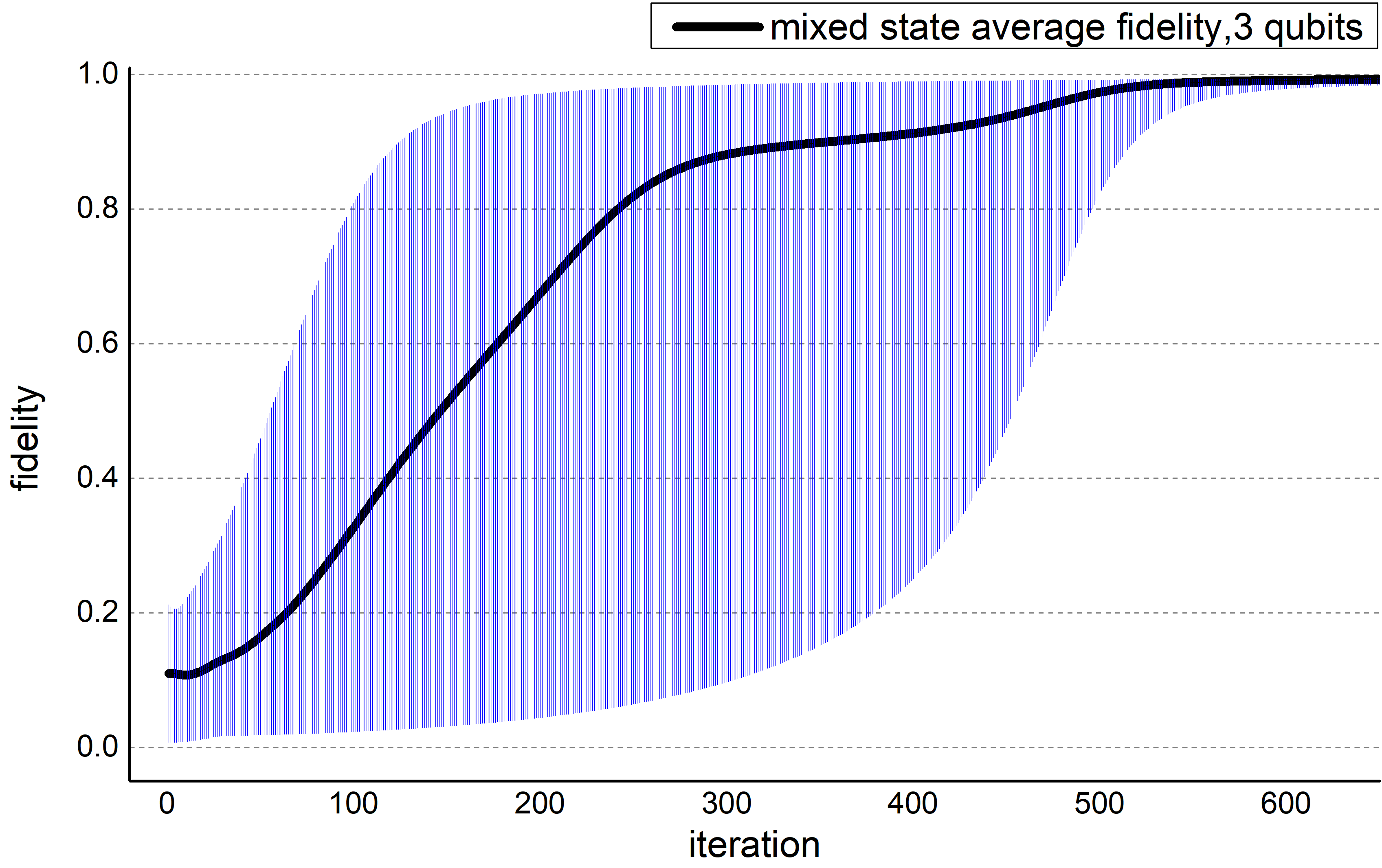}
  \caption*{3 qubits}
\end{minipage}\hfill
  \caption{Average performance of learning mixed states (1, 2, 3 qubits) where the black line is the average fidelity over multi-runs with random initializations and the shaded area refers to the range of the fidelity.}
  \label{fig:average-mixed}
\end{figure}

\begin{figure}[!tb]
  \begin{minipage}[b]{0.40\textwidth}%
  \centering
  \includegraphics[width=0.7\linewidth]{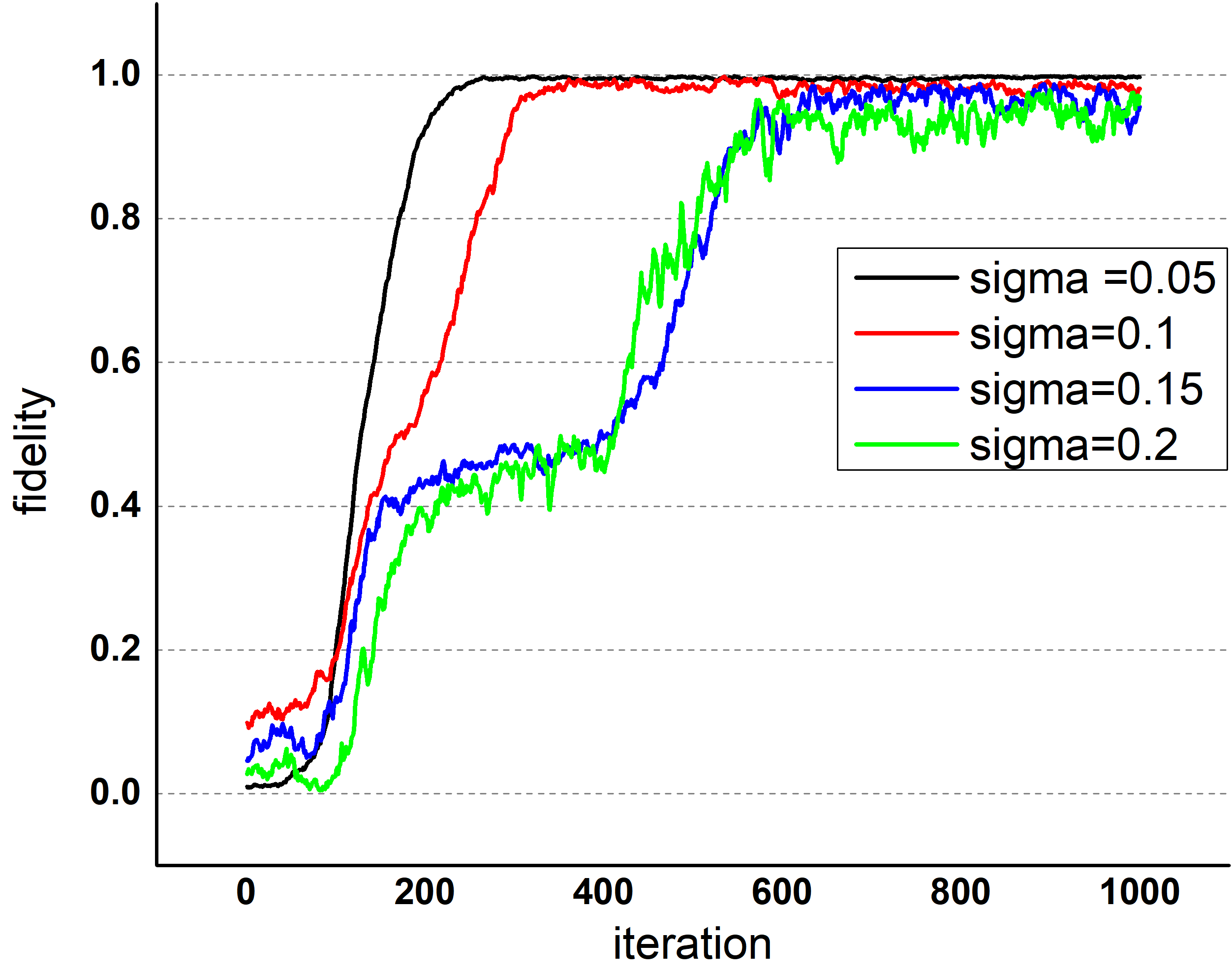}
  \caption{Learning 4-qubit pure states with noisy quantum operations.} \label{fig:noisy-learning}
  \end{minipage}\hfill
  \begin{minipage}[b]{0.55\textwidth}%
  \centering
  \includegraphics[width=0.7\linewidth, height=3cm]{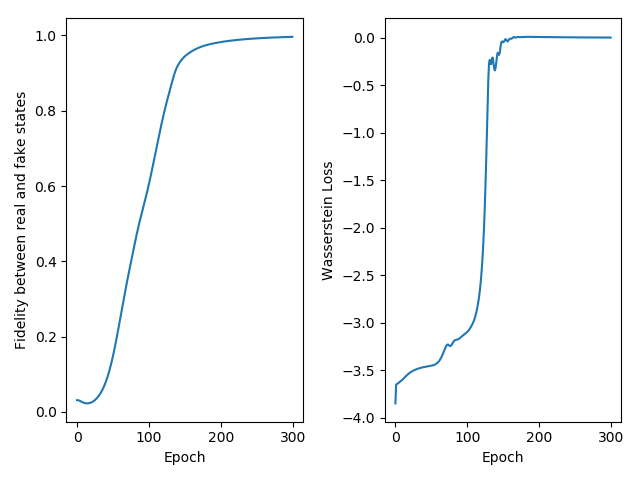}
  \caption{Learning to approximate the 3-qubit Hamiltonian simulation circuit of the 1-d Heisenberg model.}\label{fig:pi-learning}
  \end{minipage}\hfill
\end{figure}

\noindent \textbf{Mixed states\ \ } We also demonstrate a typical learning of mixed quantum states of rank $2$ with $1,2$, and $3$ qubits in \fig{average-mixed}. The generator now consists of $2$ unitary operators and 2 real probability parameters $p_1,p_2$ which are normalized to form a probability distribution using a softmax layer.

\noindent \textbf{Learning pure states with noise\ \ } To investigate the possibility of implementing our quantum WGAN on near-term machines, we perform a numerical test on a practically implementable 4-qubit generator on the ion-trap machine~\cite{ionq} with an approximate noise model~\cite{Zhu}. We deem this as the closest example that we can simulate to an actual physical experiment.
In particular, we add a Gaussian sampling noise with standard deviation $\sigma=0.2, 0.15, 0.1,0.05$ to the measurement outcome of the quantum system.
Our results (in \fig{noisy-learning}) show that the quantum WGAN can still learn a 4-qubit pure state in the presence of this kind of noise. As expected, noise of higher degrees (higher $\sigma$) increases the number of epochs before the state is learned successfully.

\noindent \textbf{Comparison with existing experimental results\ \ \ } We will compare to quantum GANs with quantum data~\cite{Killoran18qgan, benedetti2018adversarial,hu2018quantum}. It is unfortunate that there is neither precise figure nor public data in their papers which makes a precise comparison infeasible. However, we manage to give a rough comparison as follows. Ref.~\cite{Killoran18qgan} studies the pure state and the labeled mixed state case for 1 qubit. It can be inferred from the plots of their results (Figure 8.b in~\cite{Killoran18qgan}) that the relative entropy for both labels converges to $10^{-10}$ after $\sim 5000$ iterations, and it takes more than $1000$ iterations for the relative entropy to significantly decrease from $1$.
Ref.~\cite{hu2018quantum} performs experiments to learn 1-qubit pure and mixed states using a quantum GAN on a superconducting quantum circuit. However, the specific design of their GAN is very unique to the 1-qubit case. They observe that the fidelity between the fake state and the real state approaches $1$ after $220$ iterations for the pure state, and $120$ iterations for the mixed state.
From our figures, qWGAN can quickly converge for 1-qubit pure states after $150-160$ iterations and for a $1$-qubit mixed state after $\sim 120$ iterations.

Ref.~\cite{benedetti2018adversarial} studies only pure states but with numerical results up to 6 qubits. In particular, they demonstrate (in Figure 6 from~\cite{benedetti2018adversarial}) in the case of 6-qubit that the normal gradient descent approach, like the one we use here, won't make much progress at all after 600 iterations. Hence they introduce a new training method. This is in sharp contrast to our \fig{pure-state-clean} where we demonstrate smooth convergence to fidelity 1 with the simple gradient descent for 8-qubit pure states within 900 iterations.

\noindent \textbf{Application: approximating quantum circuits\ \ } To approximate any quantum circuit $U_0$ over $n$-qubit space $\X$, consider Choi-Jamio\l{}kowski state $\Psi_0$ over $\X \otimes \X$ defined as $(U_0 \otimes \I_\X) \Phi$ where $\Phi$ is the maximally entangled state $\frac{1}{\sqrt{2^n}} \sum_{i=0}^{2^n-1} \vec{e_i} \otimes \vec{e_i}$ and $\{\vec{e}_i\}_{i=0}^{2^n-1}$ forms an orthonormal basis of $\X$.
The generator is the normal generator circuit $U_1$ on the first $\X$ and identity on the second $\X$, i.e., $U_1\otimes \I$.  In order to learn for the 1-d 3-qubit Heisenberg model circuit (treated as $U_0$) in~\cite{childs2018towards}, we simply run our qWGAN to learn the 6-qubit Choi-Jamio\l{}kowski state $\Psi_0$ in \fig{pi-learning} and obtain the generator (i.e., $U_1$).
We use the gate set of single or 2-qubit Pauli rotation gates. Then $U_1$ only has 52 gates, while using the best product-formula (2nd order) $U_0$ has $\sim$11900 gates.  It is worth noting that $U_1$ achieves an average output fidelity over 0.9999 and a worst-case error 0.15, whereas $U_0$ has a worst-case error 0.001.  However,  the worst-case input of $U_1$ is not realistic in current experiments and hence the high average fidelity implies very reasonable approximation in practice.


\section{Conclusion \& Open Questions}
\label{sec:conclusion}
We provide the first design of quantum Wasserstein GANs, its performance analysis on realistic quantum hardware through classical simulation, and a real-world application in this paper. At the technical level, we propose a counterpart of Wasserstein metric between quantum data. We believe that our result opens the possibility of quite a few future directions, for example:
\begin{itemize}[leftmargin=*]
\item Can we implement our quantum WGAN on an actual quantum computer? Our noisy simulation suggests the possibility at least on an ion-trap machine.
\item Can we apply our quantum WGAN to even larger and noisy quantum systems? In particular, can we approximate more useful quantum circuits using small ones by quantum WGAN? It seems very likely but requires more careful numerical analysis.
\item Can we better understand and build a rich theory of quantum Wasserstein metrics in light of~\cite{villani2008optimal}?
\end{itemize}


\section*{Acknowledgement}
We thank anonymous reviewers for many constructive comments and Yuan Su for helpful discussions about the reference~\cite{childs2018towards}.  
SC, TL, and XW received support from the U.S. Department of Energy, Office of Science, Office of Advanced Scientific Computing
Research, Quantum Algorithms Teams program. SF received support from Capital One and NSF CDS\&E-1854532.
TL also received support from an IBM Ph.D. Fellowship and an NSF QISE-NET Triplet Award (DMR-1747426).
XW also received support from NSF CCF-1755800 and CCF-1816695.


\providecommand{\bysame}{\leavevmode\hbox to3em{\hrulefill}\thinspace}


\newpage
\begin{center}
  {\Large Supplementary Materials}
\end{center}
\appendix


\section{Preliminaries}\label{append:prelim}
\subsection{Quantum Information}
We introduce necessary quantum information backgrounds for our qWGAN.

\paragraph{Quantum states}
Quantum information can be formulated in terms of linear algebra. Given the space $\C^{d}$, its computational basis is denoted as $\{\vec{e}_{0},\ldots,\vec{e}_{d-1}\}$, where $\vec{e}_{i}=(0,\ldots,1,\ldots,0)^{\dagger}$ with the $(i+1)^{\text{th}}$ entry being 1 and other entries being 0; here `$\dagger$' denotes the complex conjugate of a vector/matrix.

\emph{Pure quantum states} with dimension $d$ are represented by unit vectors in $\C^{d}$: i.e., a vector $\vec{v}=(v_{0},\ldots,v_{d-1})^{\dagger}$ is a quantum state if $\sum_{i=0}^{d-1}|v_{i}|^{2}=1$. For each $i$, $v_{i}$ is called the \emph{amplitude} in $\vec{e}_{i}$. If there are at least two non-zero amplitudes, quantum state $\vec{v}$ is in \emph{superposition} of the computational basis, a fundamental feature in quantum mechanics.

\emph{Mixed quantum states} are probabilistic mixtures of pure quantum states. Formally, a mixed state can be written as $\sum_{k=1}^{r}p_{k}\vec{v}_{k}\vec{v}_{k}^{\dagger}$ where $p_{k}\geq 0\ \forall k\in\range{r}$, $\sum_{k=1}^{r}p_{k}=1$, and $\vec{v}_{k}$ is a pure state (i.e. $\|\vec{v}_{k}\|_{2}=1$) for all $k\in\range{r}$. Denote $\rho:=\sum_{k=1}^{r}p_{k}\vec{v}_{k}\vec{v}_{k}^{\dagger}$; $\rho$ satisfies $\rho\succeq 0$, $\Tr[\rho]=1$, and $\rho^{\dagger}=\rho$ (i.e., $\rho$ is a \emph{Hermitian matrix}). Such matrices are called \emph{density matrices}, and every mixed state is a density matrix (and vice versa).

In many scenarios, quantum states are naturally composed of two parts. This comes to the concept of \emph{bipartite quantum systems}, where a bipartite quantum state $\rho_{12}$ in $\C^{d_{1}}\otimes\C^{d_{2}}$ ($d_{1},d_{2}\in\N$) can be written as $\rho_{12}=\sum_{i}c_{i}\rho_{i,1}\otimes\rho_{i,2}$ for a probability distribution $\{c_{i}\}$ and density matrices $\{\rho_{i,1}\}$ in $\C^{d_{1}}$ and $\{\rho_{i,2}\}$ in $\C^{d_{2}}$. Since $\sum_{i}c_{i}=1$ we have $\Tr[\rho_{12}]=1$, i.e., $\rho_{12}$ is a density matrix in $\C^{d_{1}}\otimes\C^{d_{2}}$; \emph{partial trace} is defined to further characterize the properties in each separate part. Formally, the partial trace on system 1 is defined as $\Tr_{1}[\rho_{12}]:=\sum_{i}c_{i}\rho_{i,2}$, whereas the partial trace on system 2 is defined as $\Tr_{2}[\rho_{12}]:=\sum_{i}c_{i}\rho_{i,1}$.

\paragraph{Qubits}
The basic element in classical computers is one bit, whereas the basic element in quantum computers is one \emph{qubit}. Mathematically, a 1-qubit state is a state in $\C^{2}$ and can be written as $a\vec{e}_{0}+b\vec{e}_{1}$ for some $a,b\in\C$ such that $|a|^{2}+|b|^{2}=1$. An $n$-qubit state can be written as $\vec{v}_{1}\otimes\cdots\otimes\vec{v}_{n}$ where each $\vec{v}_{i}$ ($i\in\range{n}$) is a qubit state, and $\otimes$ is the Kronecker product: if $\vec{u}\in\C^{d_{1}}$ and $\vec{v}\in\C^{d_{2}}$, then $\vec{u}\otimes\vec{v}\in\C^{d_{1}}\otimes\C^{d_{2}}$ is
\begin{align}
\vec{u}\otimes\vec{v}=(u_{0}v_{0},u_{0}v_{1},\ldots,u_{d_{1}-1}v_{d_{2}-1})^{\dagger}.
\end{align}
$n$-qubit states are in a Hilbert space of dimension $2^{n}$.

\paragraph{Unitary gates}
Having the definition of quantum states, it comes to the rules of their evolution. Note that we want to keep the quantum states normalized under $\ell_{2}$-norm; in linear algebra such transformations are known as \emph{unitary transformation}. Formally, a matrix $U$ is unitary iff $UU^{\dagger}=I$.

The gates in quantum computation are always unitary gates and can be stated in the circuit model\footnote{Uniform circuits have equivalent computational power as Turing machines; however, they are more convenient to use in quantum computation.} where an \emph{$n$-qubit gate} is a unitary matrix in $\C^{2^{n}}$. A common group of unitary gates on a qubit is the \emph{Pauli gates}, where
\begin{align}\label{eq:paulis}
  \sigma_I =
  \begin{bmatrix}
      1 & 0 \\
      0 & 1
    \end{bmatrix},\
  \sigma_x =
    \begin{bmatrix}
      0 & 1 \\
      1 & 0
    \end{bmatrix},\
          \sigma_y =
          \begin{bmatrix}
            0 & -i \\
            i & 0 \\
          \end{bmatrix},\
    \sigma_z =
    \begin{bmatrix}
      1 & 0 \\
      0 & -1 \\
    \end{bmatrix};
\end{align}
note that the Pauli gates form a basis of all the unitaries acting on $\C^{2}$. Furthermore, $\sigma_I^{2}=\sigma_x^{2}=\sigma_y^{2}=\sigma_z^{2}=I$; this implies that the exponentiation of a Pauli matrix is a linear combination of Pauli matrices: for any phase $\theta\in\R$ and $\sigma\in\{\sigma_I,\sigma_x,\sigma_y,\sigma_z\}$, the Taylor expansion of $e^{\theta\sigma}$ is
\begin{align}\label{eq:Pauli-identity}
e^{\theta\sigma}=\sum_{k=0}^{\infty}\frac{\theta^{k}\sigma^{k}}{k!}=\sum_{k=0}^{\infty}\frac{\theta^{2k}}{(2k)!}I+\sum_{k=0}^{\infty}\frac{\theta^{2k+1}}{(2k+1)!}\sigma.
\end{align}

\paragraph{Quantum measurements}
Quantum states can be measured by quantum measurements. For pure states, the simplest measurement is to measure in the computational basis; for $\vec{v}=(v_{1},\ldots,v_{n})$, such measurement returns $k$ with probability $|v_{k}|^{2}$ for all $k\in\range{n}$. Recall that $\vec{v}$ is normalized such that $\|\vec{v}\|_{2}=1$, the measurement outcome constitutes a probability distribution on $\range{n}$. For $n$-qubit pure states $\vec{v}$, a common measurement is the \emph{Pauli measurement}, where you first apply $\vec{v}$ by a tensor of Pauli gates $\sigma_{1}\otimes\cdots\otimes\sigma_{n}$ ($\sigma_{1},\ldots,\sigma_{n}\in\{\sigma_I,\sigma_x,\sigma_y,\sigma_z\}$) and measure in the computational basis $\{\vec{e}_{0}\otimes\cdots\otimes\vec{e}_{0},\ldots,\vec{e}_{1}\otimes\cdots\otimes\vec{e}_{1}\}$.

For a density matrix $\rho$, the most general measurements are positive-operator valued measurements (POVMs), characterized by a set of Hermitian operators $\{E_{1},\ldots,E_{k}\}$ such that 1) $E_{i}\succeq 0$ for all $i\in\range{k}$, and 2) $\sum_{i=1}^{k}E_{i}=I$. The outcome of the measurement is $i$ with probability $\Tr[\rho E_{i}]$; this also constitutes a probability distribution as $\sum_{i=1}^{k}\Tr[\rho E_{i}]=\Tr[\rho]=1$.

\paragraph{Distance measure}
There are various of ways to define the distance between two quantum states $\rho_{1}$ and $\rho_{2}$. One natural distance is the \emph{trace distance} defined by $F_{\Tr}(\rho_{1},\rho_{2}):=\Tr|\rho_{1}-\rho_{2}|$, the sum of the absolute value of the eigenvalues of $\rho_{1}-\rho_{2}$; this generalizes the \emph{total variation distance} between classical distributions. Another common distance is the \emph{fidelity}: $F(\rho_{1},\rho_{2}):=\Tr[\sqrt{\sqrt{\rho_{1}}\rho_{2}\sqrt{\rho_{1}}}]^{2}$. $F(\rho,\sigma) = 1$ if and only if $\rho = \sigma$, and $F(\rho,\sigma)$ approaches $1$ as $\rho$ approaches $\sigma$~\cite{nielsen2002quantum}.

Besides symmetric distances, people also consider divergences as they also characterize natural properties between two distributions. One such example is the \emph{Kullback-Leibler divergence} (KL divergence)~\cite{kullback1951information}, also known as the \emph{relative entropy}, defined as follows for two classical distributions $p$ and $q$ on $\range{n}$:
\begin{align}\label{eq:classical-KL}
D_{\textrm{KL}}(p\|q)=\sum_{i=1}^{n}p_{i}\log(p_{i}/q_{i})=\sum_{i=1}^{n}p_{i}\log p_{i}-\sum_{i=1}^{n}p_{i}\log q_{i}.
\end{align}
Quantumly there is a natural extension, namely the \emph{quantum relative entropy}, defined as follows:
\begin{align}\label{eq:quantum-KL}
S(\rho\|\sigma):=\Tr[\rho(\log\rho-\log\sigma)].
\end{align}
(See \eq{matrix-log} below for the definition of $\log\rho$ and $\log\sigma$.)

To learn quantum distributions (states), one must minimize some measure of \emph{distance} between the true density matrix and our learned state; however, it turns out that the trace distance and the fidelity are not easily amenable to be optimized. This is the main reason why we adopt our quantum Wasserstein semimetric; see more discussions in \sec{qW} and \append{qW-proof}.

\paragraph{Symmetric subspace}
Recall that our quantum Wasserstein semimetric in \sec{qW} is \emph{symmetric}; achieving this requires the theory of \emph{symmetric subspaces}. Given two Hilbert spaces $\X$ and $\Y$ that are isometric, a symmetric subspace of the space $\X \otimes \Y$ is the space of those vectors that are invariant to a permutation of $\X$ and $\Y$ individually. Ref.~\cite{harrow13church} proved that the projection onto the symmetric subspace is given by
\begin{align}
\Pi_{\text{sym}}:=\frac{\I + \mathrm{SWAP}}{2}
\end{align}
where $\I$ is the identity operator and $\SWAP$ is the operator such that $\SWAP(x \otimes y) = (y \otimes x), \forall x \in \X, y \in \Y$. It is also well known that $\Pi_{\text{sym}}$ is a projector on $\X \otimes \Y$, ie. $\Pi_{\text{sym}}^2 = \Pi_{\text{sym}}$, and that $\Pi_{\text{sym}}(u \otimes u) = u \otimes u$ for all quantum states $u$.
This motivates us to choose the cost matrix $C$ in \eq{quant-wgan-primal} to be the complement of the symmetric subspace, i.e.,
\begin{align}
C:=\frac{\I - \SWAP}{2}.
\end{align}
Such choice is natural because on the one hand it ensures that $\qW(\rho,\rho) = 0$ for any quantum state $\rho$, and on the other hand it promises the symmetry of the semimetric, i.e., $\qW(\rho,\sigma) = \qW(\sigma,\rho)$ for any quantum states $\rho,\sigma$.

\subsection{Matrix Arithmetics}
Unless otherwise mentioned, the matrices we consider are \emph{Hermitian}, defined as all matrices $A$ such that $A^{\dagger}=A$. For any two Hermitian matrices $A,B\in\C^{n\times n}$, we say $A\succeq B$ iff $A-B$ is a positive semidefinite matrix (i.e., $A-B$ only has nonnegative eigenvalues), and $A\succ B$ iff $A-B$ is a positive definite matrix (i.e., $A-B$ only has positive eigenvalues).

A function of a Hermitian matrix is computed by taking summations of matrix powers under its Taylor expansion; for instance, for any Hermitian $A$ we have
\begin{align}\label{eq:matrix-exp}
\exp(A):=\sum_{k=0}^{\infty}\frac{A^{k}}{k!},
\end{align}
and for any $0\prec B\prec 2I$ we have
\begin{align}\label{eq:matrix-log}
\log(B):=\sum_{k=1}^{\infty}\frac{(-1)^{k+1}}{k}(B-I)^{k}.
\end{align}
Furthermore, we introduce two tools for matrix arithmetics that we frequently use throughout the paper. The first is a rule for taking gradients of matrix functions:
\begin{lemma}[\cite{tsuda2005matrix}]\label{lem:matrix-gradient}
Given a Hermitian matrix $W\in\C^{n\times n}$ and a function $f\colon\R\to\R$, we define the gradient $\nabla_W f(W)$ as the entry-wise derivatives, i.e., $\nabla_W f(W):=(\frac{\partial f(W)_{ij}}{\partial W_{ij}})_{i,j=1}^{n}$. Then we have
\begin{align}
\nabla_W \Tr(W\log(W)) = [\log(W) + (W)]^\dagger = \log(W) + W.
\end{align}
\end{lemma}

For exponentiations of Hermitian matrices, we use the Golden-Thompson inequality stated as follows:
\begin{lemma}[\cite{golden1965lower,thompson1965inequality}]\label{lem:Golden-Thompson}
For any Hermitian matrices $A,B\in\C^{n\times n}$, 
\begin{align}
\Tr(\exp(A+B)) \leq \Tr(\exp(A)\exp(B)).
\end{align}
\end{lemma}


\section{Properties of the Quantum Wasserstein Semimetric}\label{append:qW-proof}
\subsection{Proofs}
\begin{lemma}
  \label{lem:quant-wass-strong-dual}
  Strong Duality holds for the semidefinite program \eq{quant-wass-vanilla}.
\end{lemma}
\begin{proof}
  Note that $\pi = \P \otimes \Q$ is a feasible solution to the primal program \eq{quant-wass-vanilla}.

  Consider the solution $\psi = -\I_\Y $, $\phi = \I_\X$ for the dual program \eq{quant-wass-vanilla}. Then $  \I_\X \otimes \psi - \phi \otimes \I_\Y - C = -2\I_\X \otimes \I_\Y - C$. For any vector $v \in \X \otimes \Y$, $v^\dagger(-2\I_\X \otimes \I_\Y - C)v = -2 - v^\dagger C v \leq -2 < 0$. Therefore $\I_\X \otimes \psi - \phi \otimes \I_\Y  \prec C$ and the solution is strictly feasible. Since a strictly feasible solution exists to the dual program and the primal feasible set is non-empty, Slater's conditions are satisfied and the lemma holds \cite[Theorem 1 (1)]{watrous2013simpler}.
\end{proof}

\lem{quant-wass-strong-dual} shows that the primal and dual SDPs have the same optimal value and thus \eq{quant-wass-dual} can be taken as an alternate definition of the Quantum Wasserstein distance.

The following theorem establishes some properties of the Quantum Wasserstein distance.
  \begin{theorem}
    \label{thm:qwass-properties-main}
    $\qW(\cdot,\cdot)$ forms a semimetric over the set of density matrices $\D(\X)$ over any space $\X$, i.e., for any $\P, \Q \in \D(\X)$,
\begin{enumerate}
    \item $\qW(\P,\Q) \ge 0$,
    \item $\qW(\P, \Q)=\qW(\Q, \P)$,
    \item $\qW(\P,\Q) = 0$ iff $\P = \Q$.
\end{enumerate}
\end{theorem}

\begin{proof}
We will use the definition of $\qW(\cdot, \cdot)$ from \eq{quant-wass-vanilla} with $\Y$ being an isometric copy of $\X$.
    \begin{enumerate}[leftmargin=*]
    \item Consider the matrix $C = \frac{\I - \SWAP}{2}$. Let $\vec{u} = \sum_{i,j \in \Gamma}u_{ij} \vec{e}_i\vec{e}_j$ be any vector in $\X\otimes\Y=\C^{|\Gamma|}\otimes \C^{|\Gamma|}$. By simple calculation,
\begin{align}
\vec{u}^\dagger C\vec{u} = \sum_{i,j}u_{ij}^\ast(u_{ij} - u_{ji}) = \sum_{i \le j} (u_{ij}^\ast - u_{ji}^\ast)(u_{ij} - u_{ji}) = \sum_{i \le j} |u_{ij} - u_{ji}|^2\ge 0;
\end{align}
thus $C$ is positive semidefinite. As a result, $\Tr(\pi C) \ge 0$ for all $\pi \succeq 0$, and $\qW(\P,\Q) \ge 0$ for all density matrices $\P,\Q \in \D(\X)$.

\item This property trivially holds because of the definition in \eq{quant-wass-vanilla} is symmetric in $\P$ and $\Q$.

\item Suppose that $P = Q$ have spectral decomposition $\sum_i\lambda_i \vec{v}_{i}\vec{v}_{i}^{\dagger}$. Consider $\pi_0 = \sum_i \lambda_i (\vec{v}_{i}\vec{v}_{i}^{\dagger}\otimes \vec{v}_{i}\vec{v}_{i}^{\dagger})$. Then, $\Tr(\pi_0 C) = \Tr(\sum_i \lambda_i (\vec{v}_{i}\vec{v}_{i}^{\dagger}\otimes \vec{v}_{i}\vec{v}_{i}^{\dagger})C) = \Tr(\sum_i \lambda_i (\vec{v}_{i}^\dagger \otimes \vec{v}_{i}^\dagger) C (\vec{v}_{i} \otimes \vec{v}_{i}))$. Since $C = \frac{I - \SWAP}{2}$, $C(\vec{v}_{i} \otimes \vec{v}_{i}) = 0$ . Thus $\Tr(\pi_0 C) = 0$ and since $C$ is positive semidefinite, this must be the minimum. Thus $\qW(\P,\P) = 0$.
\qedhere
\end{enumerate}
\end{proof}

\subsection{Regularized Quantum Wasserstein Distance}
\label{append:regul-quant-wass-appendix}
The regularized primal version of the Quantum Wasserstein GAN is constructed from \eq{quant-wgan-primal} by adding the relative entropy between the optimization variable $\pi$ and the joint distribution of the real and fake states $P \otimes Q$, given by $S(\pi \Vert P \otimes Q) = \Tr(\pi \log(\pi) - \pi\log(P \otimes Q))$:
\begin{align}
  \label{eq:quant-wass-entropic}
  \min_{\pi} \quad & \Tr(\pi C) + \lambda\Tr(\pi \log(\pi) - \pi\log(P \otimes Q)) \\
  \text{s.t.} \quad &\Tr_\Y (\pi) = P, \Tr_\X (\pi) = Q, \pi \in \D(\X \otimes \Y). \nonumber
\end{align}
Here $\lambda$ is a parameter that is chosen during training, and determines the weight given to the regularizer.

To formulate the dual, we use Hermitian Lagrange multipliers $\phi$ and $\psi$ to construct a saddle point problem:
\begin{align}
  & \min_{\pi}\max_{\psi,\phi}\quad \Tr(\pi C) + \lambda\Tr(\pi \log(\pi) - \pi\log(P \otimes Q)) \nonumber \\
          & \quad\quad\quad\quad\quad\quad\quad +\Tr(\phi(\Tr_\Y (\pi) - P)) - \Tr(\psi(\Tr_\X (\pi) - Q)) \nonumber \\
  =& \min_{\pi}\max_{\psi,\phi}\quad \Tr(\pi( C + \phi \otimes \I_\Y -\I_\X \otimes \psi)) - \Tr(P\phi) \nonumber \\
  & \quad\quad\quad\quad\quad\quad\quad + \Tr(Q\psi) + \lambda\Tr(\pi \log(\pi) - \pi\log(P \otimes Q)).
\end{align}
Switching the order of the optimizations:
\begin{align}
  &\max_{\psi,\phi}\min_{\pi}\quad \Tr(\pi( C + \phi \otimes \I_\Y -\I_\X \otimes \psi)) - \Tr(P\phi) \nonumber \\
& \quad\quad\quad\quad\quad\quad\quad  + \Tr(Q\psi) + \lambda\Tr(\pi \log(\pi) - \pi\log(P \otimes Q)).
\end{align}
Solving the inner optimization problem for $\pi$ and using \lem{matrix-gradient}, we have that for the optimal $\pi$,
\begin{align}
(C + \phi \otimes \I_\Y -\I_\X \otimes \psi) + \lambda\log(\pi) + \lambda\I - \log(P \otimes Q) = 0.
\end{align}
Thus the dual optimization problem reduces to
\begin{align}
  \label{eq:quant-entrop-wass-dual}
\max_{\phi,\psi} \quad & \Tr(Q \psi) - \Tr(P \phi) - \frac{\lambda}{e}\Tr\left(\exp\left(\frac{\log(P \otimes Q) -C - \phi \otimes \I_\Y +\I_\X \otimes \psi}{\lambda} \right)\right)\\
\text{s.t.} \quad & \phi \in \H(\X), \psi \in \H(\Y). \nonumber
\end{align}
Note that the additional term in the objective of the dual cannot be directly written as the expected value of measuring a Hermitian operator. However, we can use the Golden-Thompson inequality (\lem{Golden-Thompson}) to upper bound on the objective, which can be written in terms of the expectation as
\begin{align}
  \label{eq:quant-entrop-wass-dual}
  \max_{\phi,\psi} \quad & \Tr(Q \psi) - \Tr(P \phi) - \frac{\lambda}{e}\Tr\left((P \otimes Q)\exp\left(\frac{-C - \phi \otimes \I_\Y +\I_\X \otimes \psi}{\lambda} \right)\right) \nonumber \\
  &= \max_{\phi,\psi} \E_Q[\psi] - \E_P [\phi] - \frac{\lambda}{e}\cdot\E_{P \otimes Q}\left[\exp\left(\frac{ -C - \phi \otimes \I_\Y +\I_\X \otimes \psi}{\lambda} \right)\right] \\
   \text{s.t.} \quad & \ \phi \in \H(\X),\ \psi \in \H(\Y). \nonumber
\end{align}

The regularized optimization problem has the following property:

\begin{lemma}
  Let $f \colon \D(\X) \to \R$ be defined as
  \begin{align}
    \label{eq:def-f-in-lemma}
    &\E_Q[\psi] - \E_P [\phi] - \frac{\lambda}{e}\cdot\E_{P \otimes Q}\left[\exp\left(\frac{ -C - \phi \otimes \I_\Y +\I_\X \otimes \psi}{\lambda} \right)\right] \\
   \text{s.t.} \quad & \ \phi \in \H(\X),\ \psi \in \H(\Y). \nonumber
  \end{align}
  Then $f(P)$ is a differentiable function of $P$.
\end{lemma}
\begin{proof}
  The optimization objective \eq{def-f-in-lemma} is clearly convex with respect to its parameters. Furthermore, the second derivatives are non-zero for all $\phi,\psi$, and the optimum hence is reached at a unique point. The objective function can be rewritten as
  \begin{align}
    \E_{P \otimes Q} \left(- \phi \otimes \I_\Y +\I_\X \otimes \psi - \frac{\lambda}{e}\cdot\exp\left(\frac{ -C - \phi \otimes \I_\Y +\I_\X \otimes \psi}{\lambda}\right)\right).
  \end{align}
Since $P$ and $Q$ are density matrices and are constrained to lie within a compact set, there exists a compact region $\mathbb{S}$ that is independent of $P$ (but may depend on $\lambda$) such that the maximum lies inside $\mathbb{S}$. $f(P)$ can therefore be written as $f(P) = \max g(P,\phi,\psi)$, where $\phi,\psi \in \S$, $g$ is convex, and attains its maximum at a unique point. By Danskin's theorem \cite{danskin1966theory}, the result follows.
\end{proof}


\section{More Details on Quantum Wasserstein GAN}\label{append:impl-disc-phys}
\subsection{Parameterization of the Generator}
\label{append:param-gener}

The generator $G$ is a quantum operation that maps a fixed distribution $\rho_0$ to a quantum state $P$. Two pure distributions (states with rank $1$) are mapped to each other by unitary matrices. $\rho_0$ is fixed to be the pure state $\bigotimes_{i=1}^n e_0$. If the target state is of rank $r$, $G$ can be parameterized by an ensemble $\{(p_1,U_1),\dots,(p_r,U_r)\}$ of unitary operations $U_i$, each of which is applied with probability $p_i$. Applying a unitary $U_i$ to $\rho_0$ produces the state $U_i\rho_0 U_i^\dagger$. Applying $G$ to $\rho_0$ thus produces the fake state $p_i U_i\rho_0 U_i^\dagger$.

Each Unitary $U_i$ is parameterized as a quantum circuit consisting of simple parameterized $1$- or $2$- qubit Pauli-rotation quantum gates. An $n$-qubit Pauli-rotation gate $R_\sigma(\theta)$ is given by $\exp\left(\frac{i\theta\sigma}{2}\right)$ where $\theta$ is a real parameter, and $\sigma$ is a tensor product of $1$ or $2$ Pauli matrices. Pauli-rotation gates can be efficiently implemented on quantum computers. Thus each unitary $U_i$ can be expressed as $U_i = \prod_j e^{\frac{i\theta_{i,j}\sigma_{i,j}}{2}}$.

\subsection{Parameterization of the Discriminator}
\label{append:param-discr}
The optimization variables in the discriminator are Hermitian operators, $\phi$ and $\psi$. There are two common parameterizations for a Hermitian matrix $H$:
\vspace{-1mm}
\begin{enumerate}[leftmargin=*]
\item As $U^\dagger H_0 U$, where $U$ is a parameterized unitary operator, and $H_0$ is a simpler fixed Hermitian matrix that is easy to measure. Measuring $H$ then corresponds to applying the operator $U$ and then measuring $H_0$.
\item As a linear combination $\sum_{i = 0}^{\dim(H)} \alpha_i H_i$, where $H_i$s are fixed Hermitian matrices that are easy to measure. Measuring $H$ corresponds to measuring each $H_i$ to obtain the expectation value $m_i$, and then returning $\sum_{i = 0}^{\dim(H)} \alpha_i m_i$ as the expected value of measuring $H$.
\end{enumerate}
  We choose the latter option because it allows $\xi_R$ to be conveniently approximated by a linear combination of simple Hermitian matrices. Thus $\phi$ and $\psi$ are represented by $\sum_k\alpha_kA_k$ and $\sum_l\beta_lB_l$ where $A_k,B_l$ are tensor products of Pauli matrices. The $\alpha_k$s, $\beta_l$s constitute the parameters of the discriminator.

The overall structure of the Quantum Wasserstein GAN is given in \fig{QWGAN-structure}.

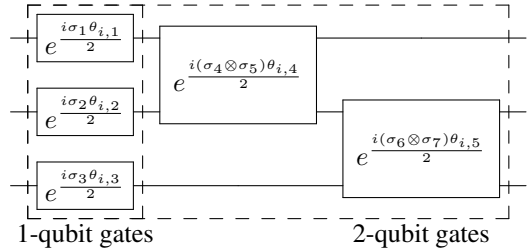
\begin{figure}[htbp]
\centering\hspace{0mm}
\Qcircuit @C=1em @R=1em {
  &\gate{e^{\frac{i\sigma_1 \theta_{i,1}}{2}}} & \multigate{1}{e^{\frac{i(\sigma_4 \otimes \sigma_5) \theta_{i,4}}{2}}} & \qw & \qw \\
  &\gate{e^{\frac{i\sigma_2 \theta_{i,2}}{2}}} & \ghost{e^{\frac{i(\sigma_4 \otimes \sigma_5) \theta_{i,4}}{2}}} & \multigate{1}{e^{\frac{i(\sigma_6 \otimes \sigma_7) \theta_{i,5}}{2}}} &\qw \\
  &\gate{e^{\frac{i\sigma_3 \theta_{i,3}}{2}}} & \qw & \ghost{e^{\frac{i(\sigma_6 \otimes \sigma_7) \theta_{i,5}}{2}}} & \qw \gategroup{1}{2}{3}{2}{0.7em}{--} \gategroup{1}{2}{3}{4}{0.7em}{--} \\
   & \mbox{1-qubit gates} & & \mbox{2-qubit gates}
}
\caption{Example parameterization of a unitary $U_i$ acting on 3 qubits. There are $12$ possible $1$-qubit gates and $48$ possible $2$-qubit gates.}
\label{fig:QWGAN-unitaries}
\end{figure}

\begin{figure}[htbp]
\centering\hspace{0mm}
\Qcircuit @C=1em @R=1em {
  \lstick{\bigotimes_{i=1}^d \vec{e}_0} & \gate{\{(p_i,U_i)\}} & \measureD{\phi} &\multicgate{3}{L} \\
  \lstick{Q} & \qw & \measureD{\psi} & \cghost{L} \\
  \lstick{\bigotimes_{i=1}^d \vec{e}_0} & \gate{\{(p_i,U_i)\}} & \multimeasureD{1}{\frac{\lambda}{e}\Tr\left(\exp\left(\frac{\log(P \otimes Q) -C - \phi \otimes \I_\Y +\I_\X \otimes \psi}{\lambda} \right)\right)} & \cghost{L}\\
  \lstick{Q} & \qw & \ghost{\frac{\lambda}{e}\Tr\left(\exp\left(\frac{\log(P \otimes Q) -C - \phi \otimes \I_\Y +\I_\X \otimes \psi}{\lambda} \right)\right)} & \cghost{L}
}
\caption{The structure of the quantum WGAN. Here $Q$ is the input state and $\vec{e}_0$ is the $0^{\text{th}}$ computational basis vector, meaning that the corresponding system is empty at the beginning. The final gate $L$ combines the outputs of the measurements of $\phi,\psi,\xi_R$ to produce the final loss function.}
\label{fig:QWGAN-structure}
\end{figure}
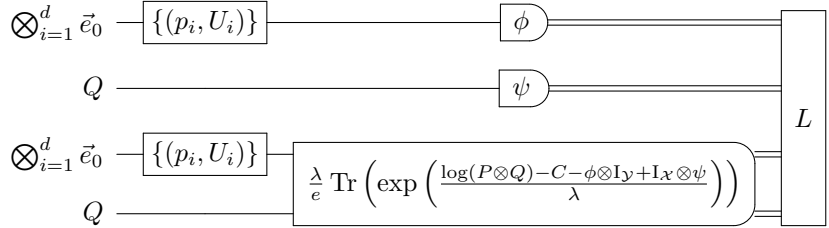

\subsection{Estimating the Loss Function}
\label{append:estim-loss-funct}
The loss function is given by $\Tr(Q\psi) - \Tr(P\phi) - \Tr((P \otimes Q)\xi_R) = \E_Q[\psi] - \E_P[\phi] - \E_{P \otimes Q}[\xi_R]$ where $\xi_R$ is the Hermitian corresponding to the regularizer term $\frac{\lambda}{e}\exp\left(\frac{-C - \phi \otimes \I_\Y +\I_\X \otimes \psi}{\lambda} \right)$.

The fake state $P$ is generated by applying a quantum operation $G$ to a fixed quantum state $\rho_0$. The quantum operation is represented by applying a set of unitary operations $\{U_1,U_2,\ldots,U_k\}$ with corresponding probabilities $\{p_1,p_2,\ldots,p_k\}$ where $k$ is the rank of the final state that would be generated:
\begin{align}
  \label{eq:1}
  P = \sum_{i \in \range{k}} p_i U_i \rho_0 U_i^\dagger.
\end{align}

\begin{lemma}
  \label{lem:lc-unitaries}
  Given a quantum state $\rho = \sum_{i=1}^k \alpha_i \rho_i$ and a Hermitian matrix $H$ then $\E_\rho(H)$ can be estimated given only the ability to generate each $\rho_i$ and to measure $H$.
\end{lemma}
\begin{proof}
  Since $\rho$ is a quantum state $\{\alpha_1,\dots,\alpha_k\}$ must form a probability distribution. Thus,
  \begin{align}
  \label{eq:2}
    \E_\rho[H] = \Tr[\rho H] = \Tr\Big[\sum_i \alpha_i \rho_i H\Big] = \sum_i \alpha_i \Tr[\rho_i H] = \sum_i \alpha_i\E_{\rho_i}[H]=\E_\alpha\E_{\rho_i}[H].
  \end{align}
  Thus we can measure the expected value of $H$ measured on $\rho$, by sampling an $i$ with probability $\alpha_i$, measuring the expected value of $H$ on $\rho_i$, and then computing the expectation over $i$ sampled from the distribution $\alpha$. We can also simply measure the expectation value $m_i$ corresponding to each $\rho_i$ and return $\sum_i \alpha_i m_i$ as the estimate.
\end{proof}

The unitaries $U_i$ are parameterized by a network of gates of the form $e^{i \theta_{i,j} \sigma_{i,j}}$ where $\sigma_{i,j}$ is a tensor product of the matrices $\sigma_x,\sigma_y,\sigma_z,I$ acting on some/all of the registers. With a sufficient number of such gates, any unitary can be represented by an appropriate choice of $\theta_{i,j}$. Since each $U_i$ is expressed as a composition of simple parameterized gates each of them can be implemented on a quantum computer and thus each $U_i\rho_0U_i^\dagger$ can be generated.

Note that $P = \sum_{i \in \range{k}} p_i U_i \rho_0 U_i^\dagger$ and $P \otimes Q = \sum_{i \in \range{k}} p_i (U_i \rho_0 U_i^{\dagger} \otimes Q)$. From \lem{lc-unitaries}, if $\phi$ and $\xi_R$ can be measured, we can estimate the terms $\E_P[\phi]$ and $\E_{P \otimes Q}[\xi_R]$. Next we show how to measure $\phi,\psi,\xi_R$ where $\phi,\psi$ are parameterized as a linear combination of tensor products of the Pauli matrices $\sigma_X,\sigma_Y,\sigma_Z,\sigma_I$.
\begin{lemma}
  \label{lem:lc-hermitians}
  Any Hermitian that is expressed as a linear combination $\sum_i\alpha_iH_i$ of Hermitian matrices $H_i$ that can be measured on a quantum computer, can also be measured on a quantum computer.
\end{lemma}
\begin{proof}
  For any fixed state $\rho$,
\begin{align}
  \label{eq:5}
  \E_\rho[H] = \Tr[\rho H] = \Tr\Big[\rho\sum_i \alpha_i H_i\Big] = \sum_i \alpha_i \Tr[\rho H_i] = \sum_i \alpha_iE_\rho[H_i].
\end{align}
Thus each of the Hermitians $H_i$ can be separately measured and the final result is the weighted average of the corresponding expectation values with coefficients $\alpha_i$.

If the $\alpha_i$ form a probability distribution, the expectation can be estimated by sampling a batch of indices from the distribution of $\alpha_i$, measuring $H_i$, and estimating the expectation averaging over the sampled indices. This procedure can be more efficient if some of the $\alpha_i$ are of very small magnitude in comparison to the others. Note that any Hermitian that can be written by as a linear combination $\sum_i\beta_i H_i$ where each $H_i$ is easy to measure can be transformed such that the coefficients form a probability distribution as $(\sum_i |\beta_i|) \sum_i \frac{|\beta_i|}{\sum_i |\beta_i|} \sgn(\beta_i) H_i$. If $H_i$ can be measured on a quantum computer, $-H_i$ can also be measured by measuring $H_i$ and negating the result.
\end{proof}

Tensor products of Pauli matrices can be measured on quantum computers using elementary techniques \cite{nielsen2002quantum}. As a result, \lem{lc-hermitians} implies that $\phi,\psi$ can be measured on a quantum computer.

Now, we prove the following lemma for expressing the regularizer term $\xi_R$:
\begin{lemma}
  \label{lem:lc-regularizer}
  The Hermitian corresponding to the regularizer term $\xi_R$ can be approximated via a linear combination of Hermitians from $\{\Sigma,\SWAP \cdot \Sigma\}$ where $\Sigma$ is a tensor product of 2-dimensional Hermitian matrices.
\end{lemma}
\begin{proof}
  Since $C = \frac{\I - \SWAP}{2}$,
\begin{align}
  \label{eq:6}
  \exp\left(\frac{ -C - \phi \otimes \I_\Y +\I_\X \otimes \psi}{\lambda} \right) = \exp\left(\frac{\SWAP - \I - 2\phi \otimes \I_\Y + 2\I_\X \otimes \psi}{2\lambda}\right).
\end{align}
Observe the following two facts:
\begin{itemize}[leftmargin=*]
\item if $\Sigma_1$ and $\Sigma_2$ are both tensor products of 2-dimensional Hermitian matrices, then $\Sigma_1 \cdot \Sigma_2$ is also a tensor product of 2-dimensional Hermitian matrices;
\item if $\Sigma$ is a tensor product of 2-dimensional Hermitian matrices, then $\SWAP \cdot \Sigma \cdot \SWAP$ is also a tensor product of 2-dimensional Hermitian matrices.
\end{itemize}
As a result, any integral power of $\SWAP - \I - 2\phi \otimes \I_\Y + 2\I_\X \otimes \psi$ can be written as a linear combination of the matrices $\{\Sigma,\SWAP\cdot\Sigma\}$ where $\Sigma$ is a tensor product of 2-dimensional Hermitian matrices. Thus any Taylor approximation of $\exp(\SWAP - \I - 2\phi \otimes \I_\Y + 2\I_\X \otimes \psi)$ is a linear combination of the same Hermitian matrices, each of which can be easily measured on a quantum computer. Thus the Taylor series for the exponential can be used to approximately measure the regularizer term.

A representation as a linear combination of the Hermitians $\{\Sigma, \SWAP \cdot \Sigma\}$, where $\Sigma$ is a tensor product of Pauli matrices, can be obtained more easily for a relaxed regularizer term
\begin{align}
  \label{eq:relaxed-regularizer}
\xi'_R=\exp\left(\frac{-C}{2\lambda}\right)\exp\left(\frac{- \phi \otimes \I_\Y +\I_\X \otimes \psi}{\lambda} \right)\exp\left(\frac{-C}{2\lambda}\right);
\end{align}
this is motivated by the Trotter formula~\cite{trotter1959product} of matrix exponentiation: for any Hermitian matrices $A,B$ such that $\|A\|,\|B\|\leq\delta\leq 1$, $\|e^{A+B}-e^{A}e^{B}\|=O(\delta^{2})$ but $\|e^{A+B}-e^{A/2}e^{B}e^{A/2}\|=O(\delta^{3})$.
Using this regularizer gives us a concrete closed form for $\xi'_R$ as a linear combination of simpler Hermitian matrices. It is less computationally intensive to compute than the original regularizer, since the only operation acting on $2n$ qubits at the same time is $\SWAP$. This relaxation also yields good numerical results in practice.

Since $(-\phi \otimes \I_\Y)(\I_\X \otimes \psi) = (\I_\X \otimes \psi)(-\phi \otimes \I_\Y) = (-\phi \otimes \psi)$, the central term in the RHS of \eq{relaxed-regularizer} is an exponential of commuting terms. If $A$ and $B$ are commuting matrices, we have $\exp(A+B) = \exp(A)\exp(B)$, and hence
\begin{align}
  \label{eq:25}
  \xi'_R = \exp\Big(\frac{-C}{2\lambda}\Big)\exp\Big(\frac{- \phi}{\lambda}\Big) \otimes \exp\Big( \frac{\psi}{\lambda} \Big)\exp\Big(\frac{-C}{2\lambda}\Big).
\end{align}

We choose $\phi$ and $\psi$ to be tensor products of terms of the form $a\sigma_x + b\sigma_y + c\sigma_z + d\I$. It can be verified that $\sigma_i\sigma_i = \I$ and $\sigma_i\sigma_j + \sigma_j\sigma_i = 2\delta_{i,j}\I$ and therefore $(a\sigma_x + b\sigma_y + c\sigma_z)^2 = (a^2 + b^2 + c^2)\I$. Given $r = \bigotimes_{i=1}^n (a_i\sigma_x + b_i\sigma_y + c_i\sigma_z + d_i\I)$, we therefore have $r^2 = \bigotimes_{i=1}^n \left(d_i(a_i\sigma_x + b_i\sigma_y + c_i\sigma_z + d_i\I) + \Pi_{i=1}^n(a_i^2 + b_i^2 + c_i^2 + d_i^2)\I\right)$ and therefore by induction,
\begin{align}
  \label{eq:9-2}
  r^k= \bigotimes_{i=1}^n\left(d_i^{k-1}(a_i\sigma_x + b_i\sigma_y + c_i\sigma_z + d_i\I) + \left(\sum_{j=0}^{k-2}d_i^j\right)(a_i^2 + b_i^2 + c_i^2 + d_i^2)\I\right).
\end{align}
Eq. \eq{9-2} can be used to expand $\exp(- \phi/\lambda) \otimes \exp(\psi/\lambda)$ using the truncated Taylor series for the exponential. Thus $\exp(- \phi/\lambda) \otimes \exp(\psi/\lambda)$ can be approximated by a linear combination of gates in $\Sigma$ up to any desired accuracy.

In addition, $C = \frac{\I - \SWAP}{2}$ implies that $C$ is a projector, i.e., $C^k = C$ for all $k\in\N^{*}$ and $C^0 = \I$. This can be used to express $\exp(C)$ in terms of only $\I$ and $C$:
\begin{align}
  \label{eq:10}
  \exp\Big(\frac{-C}{2}\Big) = \I + \sum_{j=1}^{\infty} \frac{C}{(-2)^{j}j!} = \I + \Big[\exp\Big(\frac{-1}{2}\Big) - 1\Big]C.
\end{align}
Using \eq{9-2} and \eq{10} we can compute an approximate expression (with any desired accuracy) for the relaxed regularizer $\xi'_R$ as a linear combination of the Hermitian $\{\Sigma,\SWAP\cdot\Sigma\}$ where $\Sigma$ is a tensor product of Hermitian matrices.
\end{proof}

Finally from \lem{lc-unitaries},\lem{lc-hermitians},\lem{lc-regularizer}, each of the terms $\E_Q[\psi], \E_P[\phi],\E_{P \otimes Q}[\xi_R]$ can be computed on a quantum computer.

\subsection{Direct Estimation of Gradients}
\label{append:estim-grad-directly}
In this subsection, we show how the gradients with respect to the parameters of the qWGAN can be directly estimated using quantum circuits. Suppose we have the following parameterization for the optimization variables:
\begin{align}
\rho_0 = \bigotimes_{i=1}^d \vec{e}_0\vec{e}_0^\dagger,\qquad P = \sum_{i=1}^{r} p_i U_i \rho_0 U_i^\dagger,\qquad U_i = \prod_j e^{\frac{i\theta_{i,j}H_{i,j}}{2}}
\end{align}
and
\begin{align}
  \phi = \sum_k \alpha_k A_k,\qquad\psi = \sum_l \beta_l B_l,
\end{align}
where $H_j,A_k,B_l$ are tensor products of Pauli matrices. The parameters of the generator are given by the variables $p_i, \theta_{i,j}$ and the parameters of the discriminator are given by $\alpha_k, \beta_l$. As shown in \lem{lc-regularizer}, the regularizer term $R$ can be written as $\sum_q r_qR_q$ where each $R_q$ is either a tensor product of Pauli matrices or a product of $\SWAP$ with a tensor product of Pauli matrices. Thus the loss function is given by
\begin{align}
  \label{eq:17}
  L = \Tr[Q \psi] - \Tr[P \phi] - \Tr\left[(P \otimes Q) R\right],
\end{align}
and hence
\begin{align}
  \label{eq:18}
  \frac{\partial L}{\partial p_i} = - \Tr[U_i\vec{e}_0\vec{e}_0^\dagger U_i^\dagger \phi] - \Tr\left[(U_i\vec{e}_0\vec{e}_0^\dagger U_i^\dagger \otimes Q) R\right].
\end{align}
To compute the partial derivative with respect to the parameters $p_i$, we create a fake state using only the unitary $U_i$, and compute the regularizer term as shown before:
\begin{align}
\label{eq:19} \frac{\partial L}{\partial \alpha_k} &= -\Tr[PA_k] - \Tr\left[(P \otimes Q)\frac{(A_k\otimes \I_\Y)R}{\lambda}\right]; \\
\label{eq:20} \frac{\partial L}{\partial \beta_l} &= \Tr[QB_l] - \Tr\left[(P \otimes Q)\frac{(\I_\X \otimes B_l)R}{\lambda}\right].
\end{align}

Clearly $(A_k \otimes \I_\Y)R$ and $(\I_\X \otimes B_l)R$ can be written as linear combinations of products of $\SWAP$ and tensor products of Pauli matrices, because such form exists for $A_k, B_l,R$. Thus these gradients can be measured as shown in \lem{lc-hermitians}.

Regarding the gradients with respect to $\theta_{i,j}$, we have
\begin{align}\label{eq:23}
\frac{\partial L}{\partial \theta_{i,j}} = \frac{\partial \Tr[\phi(U_i \rho_0 U_i^\dagger)]}{\partial \theta_{i,j}} - \frac{\partial \Tr[\xi_R(U_i \rho_0 U_i^\dagger \otimes Q)]}{\partial \theta_{i,j}}.
\end{align}
The terms $\frac{\partial \Tr[\phi(U_i \rho_0 U_i^\dagger)]}{\partial \theta_{i,j}},\frac{\partial \Tr[\xi_R(U_i \rho_0 U_i^\dagger \otimes Q)]}{\partial \theta_{i,j}}$ can be evaluated by modifying the quantum circuits for $U_i$ using with an ancillary control register, using previously known techniques~\cite[Section III. B]{schuld2019evaluating}. This allows us to evaluate the partial derivatives of the loss function w.r.t. the $\theta_{i,j}$ parameters.

\subsection{Computational Cost of Evaluating the Loss Function}
\label{append:comp-cost-eval}

Consider a quantum WGAN designed to learn an $n$-qubit target state with rank $r$; the generator hence consists of $r$ unitary matrices. Suppose that each unitary $U_i$ is a composition of at most $N$ fixed unitary gates. Furthermore, assume that $\phi$ and $\psi$ are parameterized as a linear combination of at most $M$ tensor products of Pauli matrices. The size of the network (the number of parameters) is thus $O(rNM)$.

The loss function consists of $3$ terms:
\begin{itemize}[leftmargin=4mm]
\item The expectation value of $\phi$ measured on the state $P$.
\item The expectation value of $\psi$ measured on the state $Q$.
\item The expectation value of $\xi_R$ measured on the state $P \otimes Q$.
\end{itemize}

The complexity of a quantum operation is quantified by the number of elementary gates required to be performed on a quantum computer. We show that a single measurement of $\phi$ on $U_i\rho_0U_i^\dagger$, $\psi$ on $Q$, and $\xi_R$ on $U_i\rho_0U_i^\dagger \otimes Q$ can be carried out using $\poly\left(n,k,N,M,\log\left(\frac{1}{\epsilon}\right)\right)$ gates.

The expectation values can then be estimated by computing the empirical expectation on a batch of measurements. These expectation values are combined as shown earlier in \append{estim-loss-funct} to obtain the expected values measured on $P$ and $P \otimes Q$.

First, $\xi_R$ can be approximated to precision $\epsilon$ via truncation of a Taylor series consisting of $\log\left(\frac{1}{\epsilon}\right)$ terms. Thus $\xi_R$ is approximated by a linear combination of $\poly\left(M,\frac{1}{\epsilon}\right)$ fixed Hermitian matrices of the form $\Sigma$ or $\SWAP \cdot \Sigma$ where each $\Sigma$ is a tensor product of 2-dimensional Hermitian matrices.

Second, by the Solovay-Kitaev theorem~\cite{dawson2005solovay}, any $n$-qubit unitary operator can be implemented to precision $\epsilon$ using $\poly\left(\log\left(n,\frac{1}{\epsilon}\right)\right)$ gates. Similarly, any fixed $n$-qubit Hermitian matrix can be measured using a circuit with $\poly\left(n,\log\left(\frac{1}{\epsilon}\right)\right)$ gates. Consequently:
\begin{itemize}[leftmargin=4mm]
\item $\psi$ can be measured on $Q$ using $M$ measurements of fixed tensor products of Pauli matrices, therefore using $\poly\left(n,M,\log\left(\frac{1}{\epsilon}\right)\right)$ gates.
\item $\phi$ can be measured on $U_i\rho_0U_i^\dagger$ for any $i$ using $M$ measurements of fixed tensor products of Pauli matrices, therefore using $\poly\left(n,M,\log\left(\frac{1}{\epsilon}\right)\right)$ gates.
\item $\xi_R$ can be measured on $U_i\rho_0U_i^\dagger \otimes Q$ for any $i$ using $\poly\left(M,\frac{1}{\epsilon}\right)$ measurements of fixed tensor products of Pauli matrices, therefore using $\poly\left(n,M,\log\left(\frac{1}{\epsilon}\right)\right)$ gates.
\item Each unitary $U_i$ can be applied by a composition of $N$ fixed unitaries, therefore using $\poly\left(n,N,\log\left(\frac{1}{\epsilon}\right)\right)$ gates.
\end{itemize}

From \append{estim-grad-directly}, it can be seen that the partial derivatives with respect to the parameters $p,\alpha,\beta$ are each computed by the same procedure as the loss function with some of the variables restricted. Furthermore, the partial derivatives with respect to $\theta_{i,j}$ can be evaluated using the circuit for $U_i$ with an ancillary register and a constant number of extra gates~\cite{schuld2019evaluating}. Each partial derivative therefore has the same complexity as the loss function. Since there are $O(rNM)$ parameters, the total gradient can be evaluated with a multiplicative overhead of $O(rNM)$ compared to evaluating the loss function.


\section{More Details on Experimental Results}
\label{append:numerical-results}

\paragraph{Pure states} We used the quantum WGAN to learn pure states consisting of $1,2,4$, and $8$ qubits. In this case, the generator is fixed to be a single unitary. The parameters to be chosen in the training are $\lambda$ (the weight of the regularizer) and $\eta_g,\eta_d$ (the learning rates for the discriminator and generator parameters, respectively). The training parameters for our experiments for learning pure states are listed in \tab{pure-states-noiseless}.
\begin{table}[H]
\begin{center}
\scalebox{0.7}{
\begin{tabular}{c|c|c|c|c}
\hline
\hline
Parameters & 1 qubit & 2 qubits & 4 qubits & 8 qubits\\
\hline
$\lambda $ & 2 & 2 & 10 & 10\\
$\eta =\eta_g=\eta_d $ & $10^{-1}$ & $10^{-1}$ & $10^{-1}$ & $10^{-2}$ \\
\hline
\end{tabular}}
\end{center}
\vspace{3mm}
\caption{Parameters for learning pure states.}\label{tab:pure-states-noiseless}
\end{table}

For 1,2, and 4 qubits, in addition to \fig{average-pure}, we also plot the average loss function for a number of runs with random initializations in \fig{average-pure-loss} which shows the numerical stability of our quantum WGAN.

\begin{figure}[!tb]
  \begin{subfigure}[b]{0.50\textwidth}
    \centering
  \includegraphics[width=5cm, height=3cm]{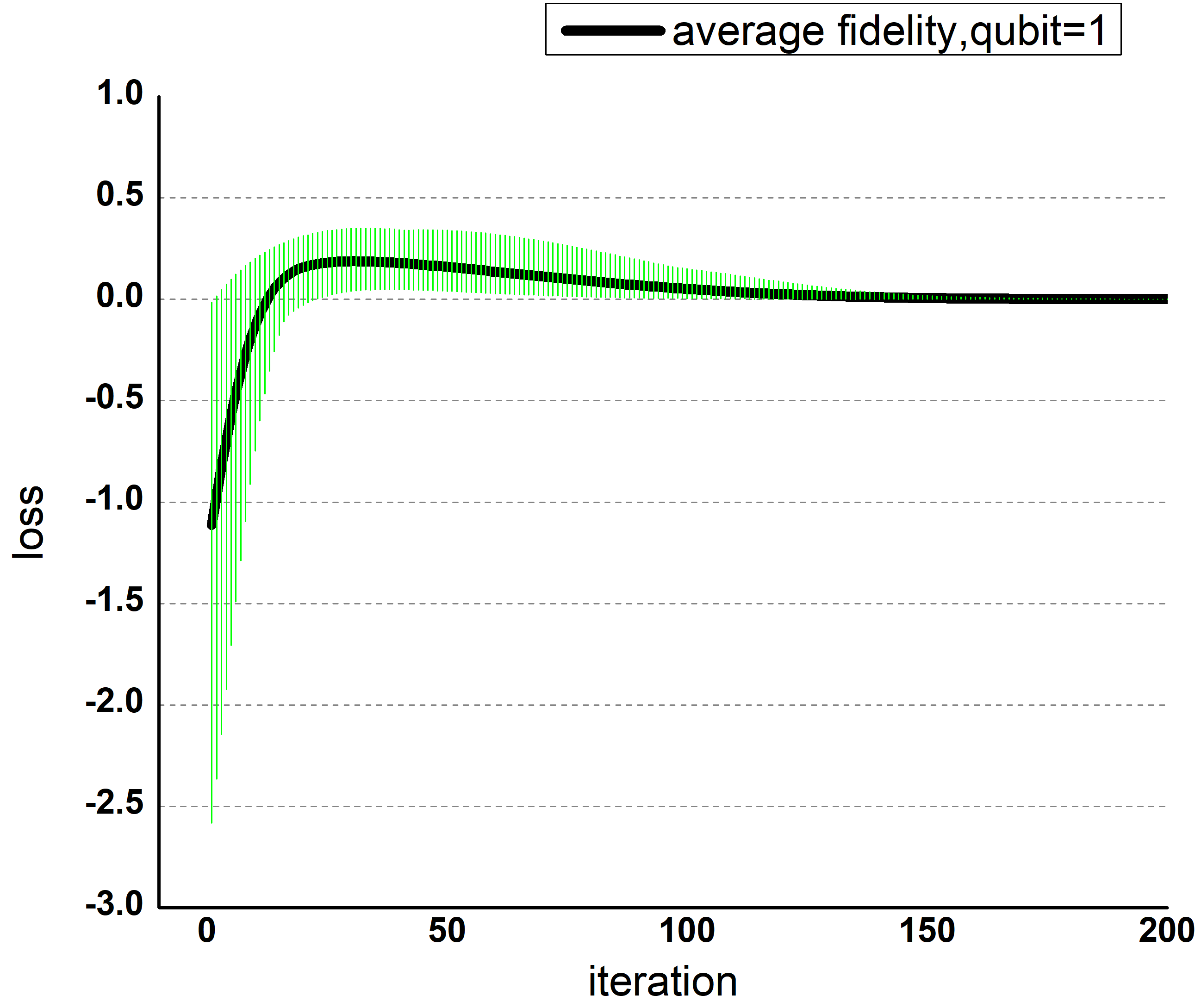}
  \caption*{1 qubit}
  \end{subfigure}\hfill
  \begin{subfigure}[b]{0.50\textwidth}
    \centering
  \includegraphics[width=5cm, height=3cm]{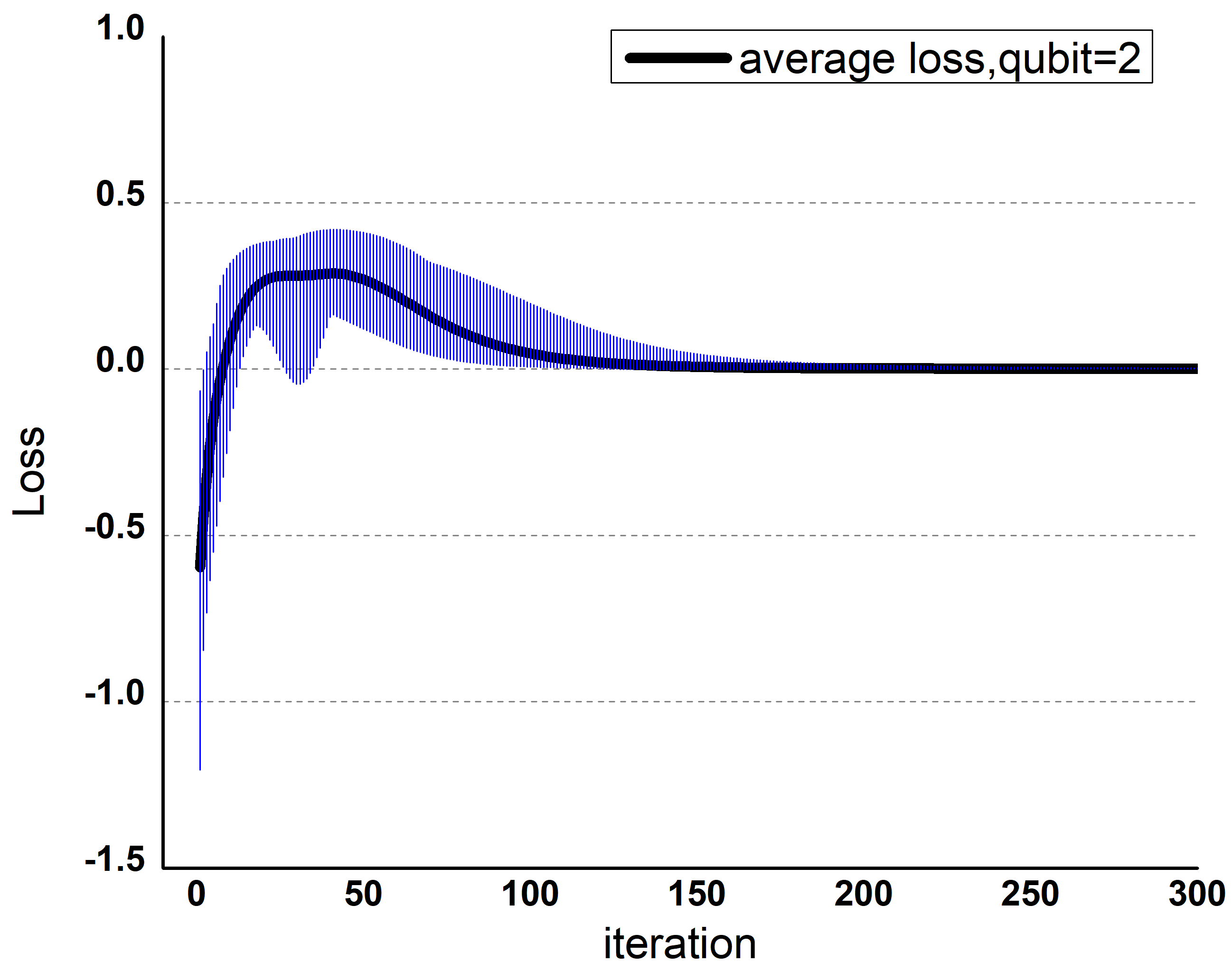}
  \caption*{2 qubits}
  \end{subfigure}\hfill
  \begin{subfigure}[b]{0.50\textwidth}%
    \centering
  \includegraphics[width=5cm, height=3cm]{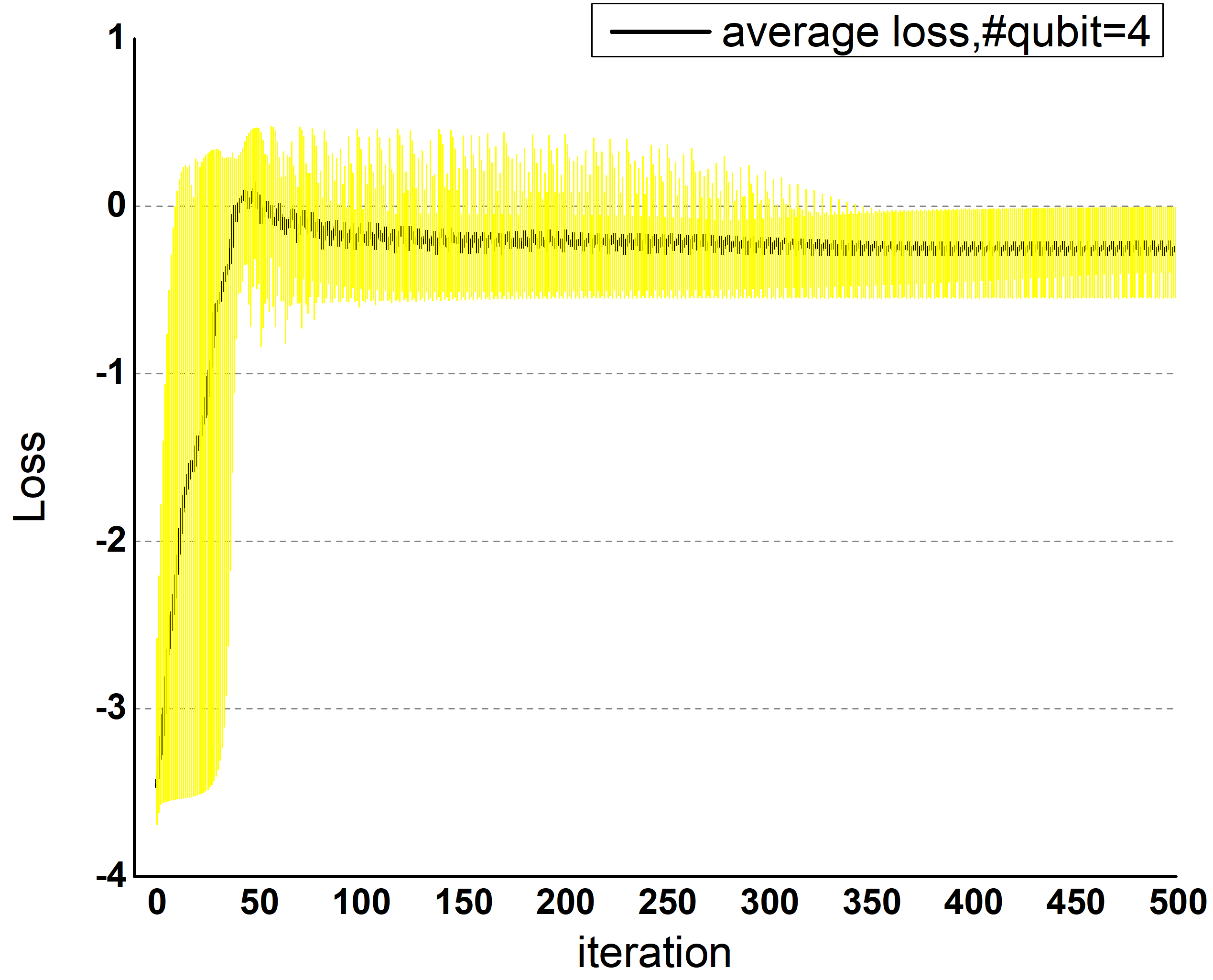}
  \caption*{4 qubits}
  \end{subfigure}\hfill
  \begin{subfigure}[b]{0.50\textwidth}%
    \centering
  \includegraphics[width=5cm, height=3cm]{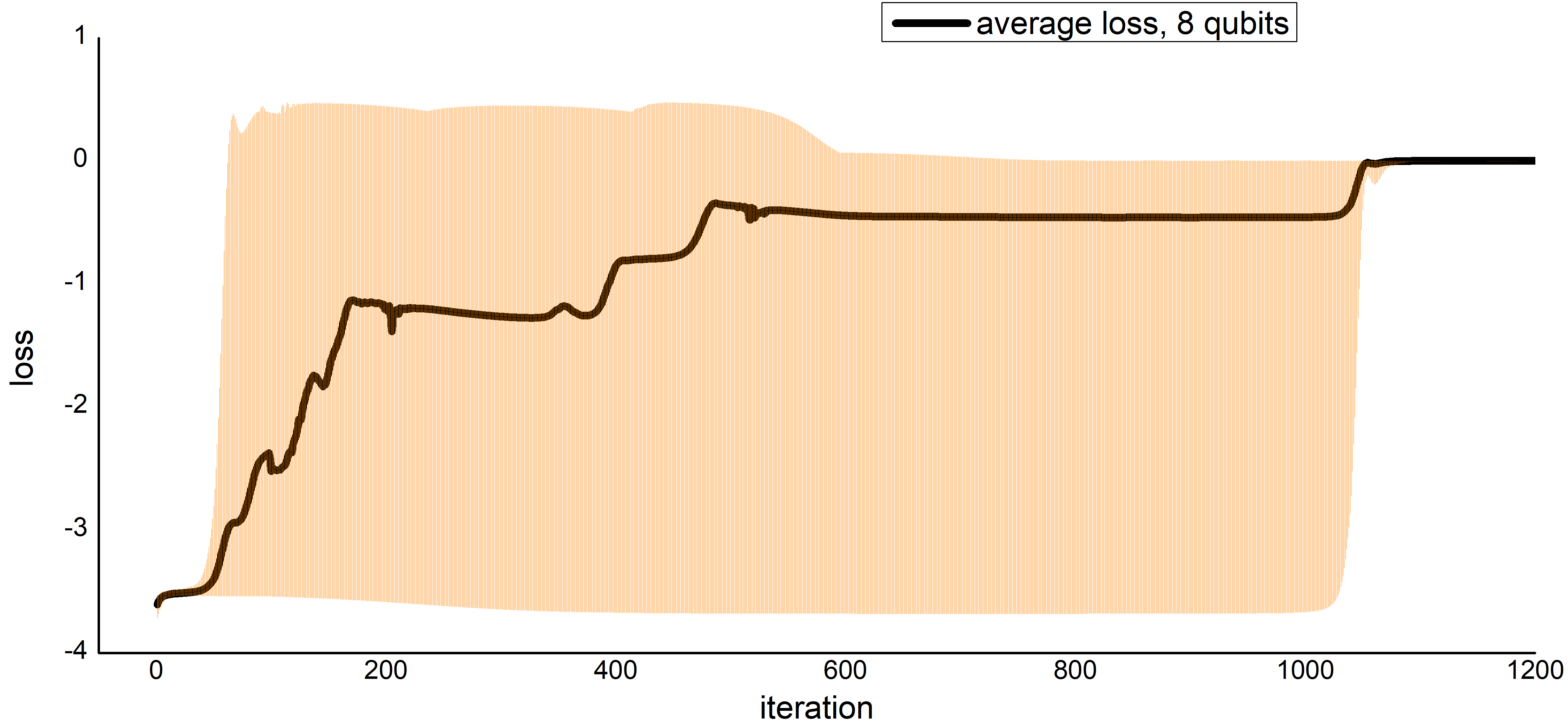}
  \caption*{8 qubits}
  \end{subfigure}
  \caption{Average performance of learning pure states (1, 2, 4 qubits) where the black line is the average loss over multi-runs with random initializations and the shaded area refers to the range of the loss.}
  \label{fig:average-pure-loss}
\end{figure}

\paragraph{Mixed states} We also demonstrate the learning of mixed quantum states of rank $2$ with $1,2$, and $3$ qubits in \fig{average-mixed}. The generator now consists of $2$ unitary operators, and 2 real probability parameters $p_1,p_2$ which are normalized to form a probability distribution using a softmax layer. The learning rate for the probability parameters is denoted by $\eta_p$. The training parameters are listed in \tab{mixed-states}.

\begin{table}[H]
\begin{center}
\scalebox{0.7}{\begin{tabular}{c|c|c|c}
\hline
Parameters & 1 qubits & 2 qubits & 3 qubits \\
\hline
\hline
$\lambda $ & 10 & 10  & 10  \\
  $\eta_{d},\eta_{g},\eta_{p}$ & $(10^{-1},10^{-1},10^{-1})$ & $(10^{-1},10^{-1},10^{-1})$ & $(10^{-1},10^{-1},10^{-1})$ \\
\hline
\end{tabular}}
\end{center}
\vspace{3mm}
\caption{Parameters for learning mixed states.}\label{tab:mixed-states}
\end{table}

\paragraph{Learning pure states with noise}
In a recent experiment result~\cite{ionq}, a quantum-classical hybrid training algorithm using the KL divergence between classical measurement outcomes as the loss function on the canonical Bars-and-Stripes data set was performed on an ion-trap quantum computer. Specifically, they use the generator in \fig{ionq-generator}. Even though the goal of~\cite{ionq} is to generate a classical distribution, we still deem it as a good example of practically implementable quantum generator to testify our quantum WGAN.

\begin{figure}[htbp]
\centering\hspace{0mm}
\[
\Qcircuit @C=1em @R=.7em {
&\gate{Z}&\gate{X}&\gate{Z}&\gate{XX}\qwx[1]&\gate{XX}\qwx[2]&\gate{XX}\qwx[3]&\qw&\qw&\qw&\qw\\
&\gate{Z}&\gate{X}&\gate{Z}&\gate{XX}&\qw&\qw&\gate{XX}\qwx[1]&\gate{XX}\qwx[2]&\qw&\qw\\
&\gate{Z}&\gate{X}&\gate{Z}&\qw&\gate{XX}&\qw&\gate{XX}&\qw&\gate{XX}\qwx[1]&\qw\\
&\gate{Z}&\gate{X}&\gate{Z}&\qw&\qw&\gate{XX}&\qw&\gate{XX}&\gate{XX}&\qw\\
}
\]
\caption{The generator circuit used in Ref.~\cite{ionq} where $Z$ stands for the $e^{i\theta\sigma_z}$ gate, $X$ stands for the $e^{i\theta\sigma_x}$ gate, and $XX$ stands for the $e^{i\theta\sigma_x \otimes \sigma_x}$ gate.}
\label{fig:ionq-generator}
\end{figure}

We use the same training parameters as in the noiseless case (\tab{pure-states-noiseless}). Furthermore, we add the sampling noise (modeled as a Gaussian distribution with standard deviation $\sigma$) which is a reasonable approximation of the noise for the ion-trap machine~\cite{Zhu}. Our results show that the quantum WGAN can still learn a 4-qubit mixed state in the presence of this kind of noise. As is to be expected, noise with higher degrees (i.e., higher $\sigma$) increases the number of epochs required before the state is learned successfully. The corresponding results are plotted in \fig{noisy-learning}.

Our finding also demonstrates the different outcomes between choosing different metrics as the loss function. In particular, some of the training results reported in~\cite{ionq} demonstrate a KL distance $<10^{-4}$ but the actual quantum fidelity is only about $0.16$.  On the other side, our quantum WGAN is guaranteed to achieve close-to-1 fidelity all the time.

\paragraph{Application: Approximating Quantum Circuits}
The quantum Wasserstein GAN can be used to approximate the behavior of quantum circuits with many gates using fewer quantum gates. Consider a quantum circuit $U_0$ over $n$ qubits. It is well known \cite{nielsen2002quantum} that there exists an isomorphism between $n$ qubit quantum circuits $U$ and quantum states $\Psi_{U}$ such that
\begin{align}
  \Psi_{U} &= \frac{1}{\sqrt{2^n}}\sum_{i=0}^{2^n - 1} (U \otimes \I)(\vec{e}_{i}\otimes \vec{e}_{i}) = \frac{1}{\sqrt{2^n}}\sum_{i=0}^{2^n - 1} (U(\vec{e}_{i})\otimes \vec{e}_{i}).
\end{align}
The quantum Wasserstein GAN can be used to learn a smaller quantum circuit $U_{1}$ such that $\Psi_{U_1}$ is close to $\Psi_{U_0}$. This can be done by setting the real state to $\Psi_{U_0}$, and using the GAN to learn to generate it using a circuit of the form $(U_1 \otimes \I)$ applied to $\frac{1}{\sqrt{2^n}}\sum_{i=0}^{2^n - 1} (\vec{e}_{i} \otimes \vec{e}_{i})$. The fidelity between $\Psi_{U_1}$ and $\Psi_{U_0}$ is given by the average output fidelity  for uniformly chosen inputs to $U_1$ and $U_0$.

We apply these techniques to the quantum circuit that simulates the evolution of a quantum system in the 1-dimensional nearest-neighbor Heisenberg model with a random magnetic field in the $z$-direction (considered in \cite{childs2018towards}). The time evolution for time $t$ is described by the unitary operator $e^{i\hat{H}t}$ with the Hamiltonian $\hat{H}$ given by
\begin{align}
  \label{eq:heisenberg-ham}
  \hat{H} = \sum_{j=1}^{n} \left(\sigma_x^{(j)}\sigma_x^{(j+1)} + \sigma_y^{(j)}\sigma_y^{(j+1)} + \sigma_z^{(j)}\sigma_z^{(j+1)} + h^{(j)}\sigma_z^{(j)}\right)
\end{align}
where $\sigma_{i}^{(j)}$ denotes the Pauli gate $\sigma_{i}$ applied at the $j^{th}$ qubit, and the $h^{(j)} \in [-h,h]$ are uniformly chosen at random.

We study the specific case with $t = n = 3$ and $h = 1$, with a fixed target error of $\epsilon = 10^{-3}$ in the spectral norm. Quantum circuits for simulating Hamiltonians that are represented as the sum of local parts, $e^{iHt} = e^{it \sum_{i=1}^{L} \alpha_j H_j }$, are obtained using $k^{th}$ order Suzuki product formulas $S_{2k}$ defined by
\begin{align}
  \label{eq:suzuki-formula}
  S_2(\lambda) &= \prod_{j=1}^{L} \exp(\alpha_j H_j \lambda /2) \prod_{j=L}^{1} \exp(\alpha_j H_j \lambda/2) \\
  S_{2k}(\lambda) &= \left[S_{2k-2}\left(p_{k}\lambda\right)\right]^2 S_{2k-2}\left(\left(1 - 4p_{k}\right)\lambda\right)^2\left[S_{2k-2}\left(p_{k}\lambda\right)\right]^2
\end{align}
where $p_k = 1/\left(4 - 4^{1/\left(2k-1\right)}\right)$ for $k \ge 1$.

We then approximate $e^{iHt}$ by $\left[S_{2k}\left(\frac{it}{r}\right)\right]^r$. Obtaining error $\epsilon$ in the spectral norm requires $r = \frac{(Lt)^{1 + 1/2k}}{\epsilon^{1/2k}}$. From \eq{suzuki-formula}, each evaluation of $S_{2k}$ requires $(2L)5^{k-1}$ gates of the form $e^{iH_j \theta}$ where $\theta$ is a real parameter. In the case of the Hamiltonian \eq{heisenberg-ham}, it is the sum of 12 terms each of which is the product of up to $2$ Pauli matrices. Thus the $k^{th}$ order formula $S_{2k}$ yields a circuit for simulating \eq{heisenberg-ham} requiring $(24)5^{k-1}\frac{(36)^{1 + 1/2k}}{0.001^{1/2k}}$ gates of the form $e^{i\theta \sigma}$ where $\sigma$ is a product of up to $2$ Pauli matrices. These are the gates used in the parameterization of our quantum Wasserstein GAN, and can be implemented easily on ion trap quantum computers. The smallest circuit is obtained using $S_2$ and requires $\sim 11900$ gates.

Using the quantum Wasserstein GAN for 6-qubit pure states, we discovered a circuit for the above task with 52 gates, an average output fidelity of $0.9999$, and a worst case error $0.15$. The worst case input is not realistic, and thus the 52 gate circuit provides a very reasonable approximation in practice.

\end{document}